\documentclass[11pt,a4paper,reqno]{amsproc}

\input{header2.sty}

\begin{document}

\title[Enhanced area law for the entanglement entropy]{How much delocalisation is needed for an\\ enhanced area law of the entanglement entropy?}

\author[P.\ M\"uller]{Peter M\"uller}
\address[P.\ M\"uller]{Mathematisches Institut,
  Ludwig-Maximilians-Universit\"at M\"unchen,
  Theresienstra\ss{e} 39,
  80333 M\"unchen, Germany}
\email{mueller@lmu.de}

\author[L.\ Pastur]{Leonid Pastur}
\address[L.\ Pastur]{B.\ Verkin Institute for Low Temperature Physics and Engineering, 
	Science Avenue 47, 61103 Kharkiv, Ukraine}
\email{pastur@ilt.kharkov.ua}

\author[R.\ Schulte]{Ruth Schulte}
\address[R.\ Schulte]{Mathematisches Institut,
  Ludwig-Maximilians-Universit\"at M\"unchen,
  Theresienstra\ss{e} 39,
  80333 M\"unchen, Germany}
\email{schulte@math.lmu.de}

\begin{abstract}
	We consider the random dimer model in one space dimension with Bernoulli disorder. For sufficiently 
	small disorder, we show that the entanglement entropy exhibits at least a logarithmically enhanced area law 
	if the Fermi energy coincides with a critical energy of the model where the localisation length diverges. 
\end{abstract}

\maketitle


\section{Introduction}

Entanglement is one of the core features of quantum mechanics, having no counterpart in classical mechanics. 
Its different facets have been the object of major research activities in various branches of modern 
physics and quantum information science \cite{Horodecki:2009gb}. 
Bipartite entanglement entropies serve as a popular quantifier of the 
degree of entanglement between two subsystems of a quantum system. Depending on the state in which the 
quantum system is prepared, entanglement entropies can show peculiar features. 
For example, Bekenstein's investigations of toy models 
for the Hawking entropy of black holes revealed \cite{PhysRevD.7.2333, Bekenstein04} 
that an entanglement entropy 
is not always an extensive quantity 
which scales with the volume but rather with the surface area of the (sub-) system. 
This seminal discovery of a so-called area law led to a wealth of further research in physics
\cite{Amico:2008en, Laflorencie:2016kg}. Soon after,
area laws for entanglement entropies were found in a variety of quantum systems that are prepared in their 
ground states, see e.g.\ \cite{RevModPhys.82.277}. 
One example concerns the entanglement entropy $S_{\Psi}$ for the spatial 
bipartition of a multi-particle or spin system in $d$ space dimensions prepared in its ground state $\Psi$. 
If $\Psi$ is energetically separated from the excited states by 
an energy gap, it is expected -- and proved for rather general one-dimensional systems 
\cite{Hastings1D07} -- that an area law will hold, $S_{\Psi} \sim L^{d-1}$. Here, $L$ is the linear 
system size.  More generally, exponential decay of ground-state 
correlations are known to be a sufficient criterion \cite{MR3296162,Stolz:2019uc}. 
Spin chains in a random magnetic field provide examples that this is fulfilled if there is a 
(suitable) mobility gap instead of a spectral gap \cite{AbdulRhamanStolz15, MR3671049, MR3749589, MR3799968, Stolz:2019uc}. 
The same has been verified for certain bosonic systems, too: randomly coupled harmonic oscillator 
systems with a mobility gap lead to an area law of the entanglement entropy \cite{MR3088230, MR3777285, 
2018arXiv181209144B}. 
On the other hand, in the absence of gaps, violations of area laws 
for entanglement entropies of ground states show up. 
They manifest themselves in a growth rate faster than the area and slower than the volume, see 
e.g.\ \cite{MR3582443}. 
Quite often, a logarithmic enhancement to the area law is found, $S_{\Psi} \sim L^{d-1} \ln L$, most 
notably if the system is probed at a quantum critical point
\cite{MR2083138, MR2566337, MR2566332}.  

Our general understanding of the scaling behaviour of entanglement entropies is far from satisfactory.
Therefore it is not only legitimate but demanded to turn to simpler or toy systems in order 
to gain further insight. Quasi-free fermion gases have proved very successful in this way from 
the perspective of mathematical physics. Their ground states $\Psi=\Psi_{E_{F}}$ are parametrised by the Fermi energy $E_{F}$, and we write $S_{E_{F}} := S_{\Psi_{E_{F}}}$ for the corresponding entanglement entropy. Several recent mathematically rigorous studies contribute 
to a first understanding of the scaling of $S_{E_{F}}$ on a more fundamental level: 
~(i)~ In case the single-particle Hamiltonian is a multi-dimensional random Schr\"odinger operator and the Fermi energy 
lies in the region of complete localisation, the validity of an area law, $S_{E_{F}} \sim L^{d-1}$, 
is established in 
\cite{PhysRevLett.113.150404,ElgartPasturShcherbina2016, MR3744386}. 	The proofs rely on the exponential decay
in space of the Fermi projection for $E_{F}$ in the region of complete localisation.
~(ii)~ Logarithmically enhanced area laws, $S_{E_{F}} \sim L^{d-1} \ln L$, are proven to occur in the case of free fermions in $d$ dimensions
\cite{Wolf:2006ek, Helling:2011gr, LeschkeSobolevSpitzer14, LeschkeSobolevSpitzer17} and if the single-particle Hamiltonian is a one-dimensional periodic Schr\"odinger operator \cite{PfirschSobolev18}. These works
even provide the exact asymptotics of $S_{E_{F}}$. In the continuum, this is achieved by relating 
it to Widom's conjecture which was proven in  celebrated works by Sobolev \cite{MR3076415, Sobolev:2014ew}.

Whereas in (i) the Fermi energy lies in a spectral region of dense pure point spectrum with corresponding
eigenfunctions that are not only exponentially localised in space but also give rise to dynamical 
localisation, the spectrum in case (ii) is absolutely continuous with delocalised generalised eigenfunctions. 
It is therefore only natural to ask the following questions: does a logarithmic enhancement to an 
area law require absolutely continuous spectrum or is a weaker breakdown of localisation 
already sufficient? Is it possible for disordered fermions to violate the area law at all?

This paper answers both questions affirmatively. To this end 
we consider the random dimer model in one dimension with Bernoulli
disorder. Its almost-sure spectrum is only pure point, but there exist two critical energies where the
localisation length diverges. Suppose the Fermi energy coincides with such a critical energy. 
Then our main result (Theorem~\ref{thm:main}) establishes a logarithmically diverging 
lower bound for the disorder-averaged entanglement entropy. Thus, this paper provides the first 
mathematical proof for the violation of an area law for a non-exactly solvable system. 
An important ingredient in our proof are the delocalisation properties -- approximate clock-spacing of
eigenvalues and flatness of eigenfunctions -- for the finite-volume random dimer model in a 
critical energy window proved by Jitomirskaya, Schulz-Baldes and Stolz \cite{JitomirskayaSchSt}. 
We combine them with a careful analysis of Pr\"ufer angles for finite-volume systems in 
Section~\ref{sec:good-cont} and an approximation argument to pass to the infinite-volume limit in 
Section~\ref{sec:ee-mod-ee}. The latter turns out to be delicate because it leads to a logarithmically growing error term
which must be dominated by the logarithmically growing main term.
Finally, we emphasise that complete localisation prevails everywhere else in the spectrum of the 
random dimer model except at the critical energies, and therefore an area law holds for the entanglement 
entropy whenever the Fermi energy does not coincide with a critical energy.

\section{Model and results}\label{ch:results}
We consider a system of quasi-free fermions whose configuration space is the one-dimensional lattice of integers $\Z$.
Its entanglement entropy of the (zero-temperature) 
ground state reduced to a spatial subset $A\subset\Z$ can be entirely expressed in terms of single-particle quantities, see e.g.\ \cite{Peschel:2003gz, Klich06, AbdulRhamanStolz15}. It is given by 
\begin{equation}
	\label{ee-def}
	S_{E_{F}}(A) := \tr{h\big(1_{A}(X) 1_{< E_{F}}(H) 1_{A}(X)\big) },
\end{equation}
where $E_{F} \in \R$ is the Fermi energy that characterises the ground state, and the trace is over the single-particle Hilbert space $\ell^{2}(\Z)$ of complex-valued square-summable sequences indexed by $\Z$. The function $h:[0,1]\rightarrow\R_{\ge0}$ is given by 
\begin{equation}
	\label{gdef}
	h(\lambda):=-\lambda \log_{2}\lambda - (1-\lambda)\log_{2}(1-\lambda) 
\end{equation}
for $\lambda\in[0,1]$ with the convention $0\log_{2} 0 :=0$. We write $1_{M}$ for the indicator function of a set $M$ and, in abuse of notation, $1_{< E_{F}} := 1_{]-\infty, E_{F}[}$. Finally, $H$ denotes the single-particle Hamiltonian and $X$ the position operator. Our particular interest lies in the case where $H$ is 
an operator-valued random variable.

The Hamiltonian $H: \Omega \ni \omega \mapsto H^\omega$ of the random dimer model is given by the sum of the kinetic part represented by the discrete Laplacian and a random potential, 
\begin{equation}
	\label{RSO-def}
	H^\omega:=-\sum_{x\in\Z} \big(\proj{\delta_{x}}{\delta_{x+1}}+\proj{\delta_{x+1}}{\delta_{x}} \big) + v \sum_{x\in\Z}V^\omega(x)\proj{\delta_{x}}{\delta_{x}}.
\end{equation}
Here, $(\Omega, \mathcal{A},\mathbb{P})$ is a probability space and, for a given disorder configuration $\omega$, the realisation $H^{\omega}$ acts as a bounded linear operator on $\ell^{2}(\Z)$. We write 
	$\{\delta_{x}\}_{x\in\Z}$ for the canonical basis of $\ell^{2}(\Z)$ and use the Dirac notation for rank-1 operators.
The random potential with disorder strength $v>0$ acts as the multiplication operator by the single-site potentials $\big(V^{\omega}(x)\big)_{x\in\Z}$, which are the realisations of a family of real-valued random variables with the properties $V(2x)=V(2x+1)$
for all $x\in\Z$ and $\big(V(2x)\big)_{x\in\Z}$ are independently and identically distributed. This means that every other pair of consecutive sites shares the same value of the potential. The random variable $V(0)$ is Bernoulli distributed. It attains the two different potential values $V_\pm \in\R$ with probability $p_{\pm}\in\;]0,1[\,$, subject to $p_{+} + p_{-} = 1$.  Without loss of generality, we set $V_-:=0$ and $V_+:=1$. The random Schr\"odinger operator $H$ describes a random infinite sequence of two kinds of homodimers linked together to an infinite chain.
Standard ergodicity arguments \cite{carlac1990random, pastfig1992random, AizWarBook} -- here with respect to $2\Z$-translations -- imply that the spectrum of the operator \eqref{RSO-def} is given by $\sigma(H^\omega)=[-2,+2]\cup[-2+v,2+v]$ for $\mathbb{P}$-almost every $\omega \in\Omega$.

The most interesting property of this model is that, although it is a one-dimensional discrete model, it does exhibit characteristics of delocalisation at the isolated critical energies $\{0,v\}$ in the spectrum, as was proven in \cite{JitomirskayaSchSt}, provided $v<2$. For the convenience of the reader, we state the precise result in the next section. These points of delocalisation can only occur at energies where the Lyapunov exponent vanishes. There exist further possibilities for a vanishing Lyapunov exponent in this model according to \cite{DeBievre:2000tq}. However, it is not clear what kind of transport to expect at these other energies. In any case, \cite{DeBievre:2000tq} prove strong dynamical localisation away from all these exceptional energies. 

Our main result shows the presence of at least a logarithmic enhancement to the area law of the disorder-averaged entanglement entropy. It pertains to the zero-temperature ground state of the non-interacting fermion system with  single-particle Hamiltonian given by  \eqref{RSO-def}. The Fermi energy is critical and the disorder strength sufficiently weak. Given $L\in\N$, let $\Lambda_{L} :=\{1,\ldots,L\}$ be 
	a box in $\Z$ consisting of $|\Lambda_{L}| =L$ sites.

\begin{theorem}
	\label{thm:main}
	Consider the entanglement entropy \eqref{ee-def} for the Hamiltonian \eqref{RSO-def} of the 
	random dimer model. Then, there exists a maximal disorder strength $v_{0} \in \;]0,2[$ such that 
	for every $v\in \;]0, v_{0}]$ and for a critical Fermi energy $E_F\in\{0,v\}$, we have
	\begin{equation}
		\label{eq:main}
		\liminf_{L \to\infty} \, \frac{\E{S_{E_{F}}(\Lambda_L)}}{\ln L} >0.
	\end{equation}
	Here, $\mathbb{E}$ denotes the expectation corresponding to the probability measure $\mathbb{P}$.
\end{theorem}

In proving the theorem, we obtain, as an intermediate result, an enhancement to the area law for a modified entanglement entropy. The modification consists in replacing the infinite-volume Hamiltonian \eqref{RSO-def} in \eqref{ee-def} by its simple restriction $H_L^\omega$ to the box $\Gamma_{L} :=\{-L,\hdots,L-1\} \subset\Z$. For $A \subset \Gamma_{L}$, we thus define this modified entanglement entropy as
\begin{equation}
	\label{mod-ee-def}
	S^{\omega}_{E_{F}}(A,{L}) := \tr{h\big(1_{A}(X) 1_{< E_{F}}(H^{\omega}_{L}) 1_{A}(X)\big)}
\end{equation}
for $\mathbb{P}$-almost every $\omega\in\Omega$.


\begin{theorem}
	\label{thm:GeneralExpectationLogLowerBound}
	Let $v\in \;]0,2[$ and the Fermi energy $E_F\in\{0,v\}$ be critical.
	Then there exists $\delta_{0}\in\;]0,1[$ such that for all $\delta\in\;]0,\delta_{0}]$ the modified 
	entanglement entropy 
	satisfies
	\begin{equation}
		\label{eq:thm2}
		\liminf_{L\to\infty} \, \frac{S^{\omega}_{E_{F}}(\Lambda_{L}',L)}{\ln L} >0
	\end{equation}
	for $\mathbb{P}$-almost all $\omega\in\Omega$. Here, we have defined
	$\Lambda_{L}' := [-L, -(1-\delta) L] \cap \Z$.
\end{theorem}


\begin{remarks}
\item The proof of Theorem~\ref{thm:main} shows that the left-hand side of \eqref{eq:main} 
	is bounded from below by $ 2^{-16}$, see \eqref{eq:main-term-est} and \eqref{calL-cond}.
	More interestingly, the proof of Theorem~\ref{thm:GeneralExpectationLogLowerBound} yields a
	strictly positive constant, which depends only on $v$, but not on $\omega$ that serves as a lower 
	bound for the limit inferior in \eqref{eq:thm2}.	 
\item We point out that, in contrast to Theorem \ref{thm:main}, the validity of 
	Theorem~\ref{thm:GeneralExpectationLogLowerBound} is not restricted to weak disorder. Furthermore, 
	it provides an almost-sure bound, whereas Theorem \ref{thm:main} holds in expectation. This is of relevance, because 
	the entanglement entropy is known \emph{not} to be self-averaging in one dimension \cite{MR3744386}.
	The price we pay is that the box $\Lambda'_L$ is attached to one boundary point of $\Gamma_L$. 
	Our methods in Section \ref{sec:ee-mod-ee} do not allow us to pass to the (non-modified) entanglement 
	entropy in this situation.
\item Modified entanglement entropies with boxes attached to a boundary as in \eqref{eq:thm2} are often considered in
	physics, see e.g.\	 \cite{PhysRevB.85.094417, PhysRevB.89.115104}.
\item
	A divergent lower bound for the modified entanglement entropy of a deterministic system of coupled harmonic oscillators was recently proven in \cite{2018arXiv181209144B}.
\item For all energies at which the Lyapunov exponent does not vanish, the multi-scale analysis can 
	be applied to prove (strong dynamical) localisation, despite the Bernoulli distribution of the random variables
	\cite{CarmonaKleinMar87, DeBievre:2000tq}. Some additional work then yields fast decay of the Fermi projection 
	at all these energies. Thus, it follows from  \cite{PhysRevLett.113.150404,ElgartPasturShcherbina2016} 
	that the entanglement entropy exhibits an area law at all non-critical Fermi energies of the random dimer model.
\end{remarks}


\section{Proof of Theorem~\ref{thm:GeneralExpectationLogLowerBound}}

\subsection{General idea and strategy}

The aim is to construct a suitable lower bound on the expectation of the 
modified entanglement entropy $S_{E_{F}}(A,L)$ from \eqref{mod-ee-def} that grows logarithmically in $L$. 
In doing so, we choose a subregion $\Lambda_L \subset \Gamma_L$ for $A$ whose length is proportional to 
$L$ and which is carefully positioned within $\Gamma_{L}$. This will allow us to control the error 
when passing to the (non-modified) entanglement entropy \eqref{ee-def} in the next section. 

A typical first step \cite{PhysRevLett.113.150404} in obtaining a lower bound for the (modified) 
entanglement entropy is to replace the function $h$
in its definition by a parabola.

\begin{definition}
	\label{def:quadratic-ent}
 	Let $g: [0,1] \rightarrow \R_{\ge 0}$,
	\begin{equation}
 		\label{hdef}
		g(\lambda) :=  4 \lambda(1-\lambda),
	\end{equation}	
	so that $g\le h$. For $E_{F}\in\R$ and $A\subset\Z$, respectively $A\subset\Gamma_{L}$, we introduce the quadratic analogue to the (modified)	entanglement entropy
	\begin{equation}
		\begin{split}
 			Q_{E_{F}}^{\omega}(A) &:= \tr{g\big(1_{A}(X) 1_{< E_{F}}(H^{\omega}) 1_{A}(X)\big) } 
			\le S^{\omega}_{E_{F}}(A), \\
			Q_{E_{F}}^{\omega}(A,L) &:= \tr{g\big(1_{A}(X) 1_{< E_{F}}(H^{\omega}_{L}) 1_{A}(X)\big)} \le 
				S^{\omega}_{E_{F}}(A,{L}),
		\end{split}
	\end{equation}
	$\omega\in\Omega$, of the random dimer model.
\end{definition}

Since the finite-volume Schr\"odinger operator has only discrete spectrum, $Q_{E_{F}}^{\omega}(A,L)$
can be conveniently rewritten in terms of the non-degenerate eigenvalues 
$E \in \sigma(H_{L}^{\omega})$ and corresponding $\ell^{2}(\Gamma_{L})$-normalised eigenfunctions 
$\psi_{E}^{\omega}$ (for which we drop the index $L$ from the notation). For convenience, we set $\psi_{E}^{\omega}(-L-1) := 0 =: \psi_{E}^{\omega}(L)$.

%
%

\begin{lemma}\label{prop:Summe}
	Let $x_{1},x_{2} \in\Z$ with $-L \le  x_1 < x_2 \le L-1$, $A:= [x_{1},x_{2}] \cap\Z$ and $E_{F} \in\R$.  
	Then we have
	\begin{equation}
		\label{eq:EntropyEnergySum}
			Q^{\omega}_{E_{F}}(A,L) 
			= 4\sum_{\substack{E,E' \in \sigma(H_{L}^{\omega}): \\[.5ex] E < E_{F}, \,E' \ge E_F}}
			\frac{1}{(E' -E)^2} \; |\langle\psi_{E}^{\omega},[H_L^{\omega}, 1_{A}(X)]\psi_{E'}^{\omega}\rangle|^2
	\end{equation}
	for $\mathbb{P}$-a.a.\ $\omega\in\Omega$, where the commutator is a boundary operator
	\begin{equation}
		[H_L^{\omega},1_{A}(X)]= \proj{\delta_{x_{1}}}{\delta_{x_{1}-1}} - \proj{\delta_{x_{1}-1}}{\delta_{x_{1}}}
			+ \proj{\delta_{x_{2}}}{\delta_{x_{2}+1}} - \proj{\delta_{x_{2}+1}}{\delta_{x_{2}}}
	\end{equation}
	that is independent of randomness.
\end{lemma}

\begin{proof}
	We introduce the abbreviations 
	$P:=1_{A}(X)$ and $Q:=1_{< E_{F}}(H^{\omega}_{L})$. A straightforward calculation of the trace yields
	\begin{align}
		\frac{1}{4}\tr{g(PQP)} &= \tr{PQP(1-Q)P} \notag \\
		&= \sum_{\substack{E,E' \in \sigma(H_{L}^{\omega}): \\[.5ex]
			 E < E_{F}, \,E' \ge E_F}}
			 \tr{P\proj{\psi_{E}^{\omega}}{\psi_{E}^{\omega}}P\proj{\psi_{E'}^{\omega}}{\psi_{E'}^{\omega}}P} 
			 \notag\\
		&= \sum_{\substack{E,E' \in \sigma(H_{L}^{\omega}): \\[.5ex]
			 E < E_{F}, \,E' \ge E_F}}
		 	 |\langle\psi_{E}^{\omega}, P\psi_{E'}^{\omega}\rangle|^2.
	\end{align}
	The matrix elements of $P$ can be rewritten in terms of the commutator according to 
	\begin{equation}
 		E \langle\psi_E^{\omega}, P\psi_{E'}^{\omega}\rangle
		= \langle\psi_{E}^{\omega}, H_L^{\omega}P\psi_{E'}^{\omega}\rangle
		= E' \langle\psi_{E}^{\omega}, P\psi_{E'}^{\omega}\rangle 
			+ \langle\psi_E^{\omega}, [H_L^{\omega},P]\psi_{E'}^{\omega}\rangle.
	\end{equation}
	This concludes the proof.                                                                                                                                                                                                                                                                                                                                                                                                                                                                                                                                                                                                                                                                                                                                                                                                                                                                                                                                                                                                                                                                                                                                                                                                                                                                                                                                                                                                                                                                                                                                                                                                                                                                                                                                                                                                                                                                                                                                                                                                                                                                                                                                                                                                                                                                                                                                                                                                                                                                                                                                                                                                                                                                                                                                                                                                                                                                                                                                                                                                                                                                                                                                                                                                                                                                                                                                                                                                                                                                                                                                                                                                                                                                                                                                                                                                                                                                                                                                                                                                                                                                                                                                                                                                                                        
\end{proof}

\begin{remark}
 	It is clear that \eqref{eq:EntropyEnergySum} holds for general self-adjoint operators $H$ 
	on $\ell^{2}(\Gamma_{L})$.
\end{remark}

The overall idea of our argument is that the energy denominator in 
\eqref{eq:EntropyEnergySum} provides the mechanism 
for the potential logarithmic enhancement to the area law. The enhancement can only occur if eigenfunctions
corresponding to near-by energies have a significant overlap somewhere on the surface of the spatial region 
$A$. For Anderson-localised systems, this is typically not the case, because the localisation centres of 
two eigenfunctions are expected to be separated by a distance that grows logarithmically with the inverse
of their energy difference \cite{Mott68PhilMag, Mott70PhilMag, MR2025824}.
Consequently, the entanglement entropy is expected to obey a strict 
area law for localised systems. 
Indeed, this was proven in \cite{PhysRevLett.113.150404, ElgartPasturShcherbina2016}, who follow another line of 
reasoning. In the dimer model, however, localisation breaks down at the critical energies and delocalisation
properties occur in an energy window around the critical energies. In fact, delocalisation manifests 
itself almost as nicely as for the Laplacian \cite{JitomirskayaSchSt}.  

In order to formulate this we introduce some more notation. Given $E\in\R$ and $\omega\in \Omega$, let $\phi_{E}^{\omega}: \Z \rightarrow \R$ be the unique 
solution of the difference equation
\begin{equation}
	\label{diff-eq}
	- \phi_{E}^{\omega}(x-1) - \phi_{E}^{\omega}(x+1) + vV^{\omega}(x) \phi_{E}^{\omega}(x) 
	= E \phi_{E}^{\omega}(x), \qquad x \in \Z,
\end{equation}
subject to the constraints $\phi_{E}^{\omega}(-L-1)=0$, $\phi_{E}^{\omega}(-L) > 0$ and 
$\sum_{x\in\Gamma_{L}} \phi_{E}^{\omega}(x)^{2}=1$. 
We write 
\begin{equation}
	\label{gen-ef}
 	\psi_{E}^{\omega} :=\phi_{E}^{\omega}\big|_{\Gamma_{L}} 
\end{equation}
for its restriction onto $\Gamma_{L}$.
This function is an eigenfunction of $H_{L}^{\omega}$ if and only if the boundary condition 
$\phi_{E}^{\omega}(L)=0$ holds also at the right border of $\Gamma_{L}$. Thus, \eqref{gen-ef} generalises our 
previous notation of eigenfunctions.

\begin{theorem}[Jitomirskaya, Schulz-Baldes, Stolz \cite{JitomirskayaSchSt}]
	\label{thm:Jitomirskaya}
	Let $v\in \; ]0,2[$ and $E_{c} \in \{0,v\}$. Then
	\begin{nummer}
	\item
	For every $\alpha>0$ there exist a minimal length $L_{\mathrm{min}} \in \N$ and quantities $c>0$ and $C>1$,
	depending on $\alpha$ and on the system parameters, with 
	\begin{equation}
		\label{eq:Cnach1}
 		\lim_{v\downarrow 0} C = 1,
	\end{equation}
	such that for all $L \ge L_{\mathrm{min}}$ there are exceptional events 
	$\Omega_L(\alpha)\subseteq\Omega$ of small probability
	\begin{equation}
		\label{prob-bad-event}
		\Pp[\Omega_L(\alpha)] \le \e^{-cL^\alpha}
	\end{equation}
	such that for every non-exceptional $\omega\in(\Omega_L(\alpha))^c$ the following statement holds:
	the eigenvalues of $H_{L}^{\omega}$ in the critical energy window 
	$\mathcal{W}_{L} := [E_c-L^{-1/2-\alpha},E_c+L^{-1/2-\alpha}]$ are 
	equally spaced in the sense that any two adjacent eigenvalues $E$ and $E^\prime$ 
	satisfy
	\begin{equation}
		\label{eq:LevelSpacing}
		\frac{\pi}{C^{3}L}\leq \left|E-E^\prime\right|\leq\frac{\pi C^{3}}{L}.
	\end{equation}
	Furthermore, any solution $\psi_{E}^{\omega}$ of \eqref{diff-eq}, defined as in \eqref{gen-ef} and 
	with energy $E\in \mathcal{W}_L$,
	is evenly spread over $\Gamma_L$ in the sense that
	\begin{align}\label{eq:GleichVerteilung}
		\frac{1}{CL}\leq\psi^{\omega}_E(x-1)^{2} + \psi^{\omega}_E(x)^{2}\leq\frac{C}{L}
	\end{align}
	for all $x=-L+1,\ldots,L-1$.
	\item
		\label{cor:dosbounds}
		The density of states $\mathcal{N}'(E_{c})$ is well defined and obeys the estimate  
		\begin{equation}
			\frac{1}{2\pi C^{3}} \le\mathcal{N}'(E_{c}) \le \frac{C^{3}}{2\pi}.
		\end{equation}
	\end{nummer}
\end{theorem}

%

\begin{remarks} 
	\item	
		Our formulation of Theorem~\ref{thm:Jitomirskaya}(i) is a slight improvement of the original theorem in
		\cite{JitomirskayaSchSt} concerning the quantity $C$. In fact, the statement \eqref{eq:Cnach1} on its limit for weak disorder is not provided by 
		\cite{JitomirskayaSchSt}. However, we need $C$ to be sufficiently close to one in our		
		proof of the enhanced area law. It is plausible that weak disorder should lead to a value of $C$ close to one 
		because the deviation of $C$ from one encodes the aberration of the level spacing from perfect clock behaviour and of the 
		flatness of the eigenfunctions, which are both found for the Laplacian.
		In order to derive \eqref{eq:Cnach1} we repeat some arguments of \cite{JitomirskayaSchSt} in Appendix~\ref{ch:Appendix}
		while carefully tracking the occurring constants. In particular, this requires additional estimates which were not 
		needed in \cite{JitomirskayaSchSt}.
	\item		
		The explicit two-sided bound on the density of states in Part~(ii) is not contained in \cite{JitomirskayaSchSt} either.
		Its proof is also contained in Appendix~\ref{ch:Appendix}.
\end{remarks}

In order to see how the logarithmic enhancement emerges we introduce Pr\"ufer variables  
$r^\omega_x(E) \in [0, \infty[$ and $\theta^\omega_x(E) \in\R$ as the polar coordinates of the pair
\begin{equation}
	\left(\begin{array}{@{}c@{}} \phi^\omega_{E}(x) \\[0.5ex] \phi^\omega_{E}(x-1) \end{array}\right) 
	=: r^\omega_x(E) \left(\begin{array}{@{}c@{}} \cos\big(\theta^\omega_x(E)\big) \\[0.5ex]
			\sin\big(\theta^\omega_x(E)\big) \end{array}\right)
\end{equation}
for every $x \in\Z$. For ease of notation, we do not keep track of the $L$-dependence of the Prüfer variables. 
The matrix element of the commutator in \eqref{eq:EntropyEnergySum} can thus be expressed as 
\begin{align}
	\label{eq:Restdarstellung1}
	\langle\psi_{E}^{\omega}, [H^{\omega}_L, 1_{A}(X)]\psi^{\omega}_{E'} \rangle
	&= r^{\omega}_{x_{2}+1}(E)r^{\omega}_{x_{2}+1}(E')\sin\big(\theta^{\omega}_{x_{2}+1}(E)-\theta^{\omega}_{x_{2}+1}(E')\big) \notag \\
	& \quad -r^{\omega}_{x_{1}}(E)r^{\omega}_{x_{1}}(E')\sin\big(\theta^{\omega}_{x_{1}}(E)-\theta^{\omega}_{x_{1}}(E')\big).
\end{align}
To construct a lower bound on $\mathbb{E}[S_{E_{F}}(A,L)]$ for $E_{F}=E_{c}$, we restrict the expectation to the event $(\Omega_{L}(\alpha))^{c}$ from Theorem~\ref{thm:Jitomirskaya} and then restrict  the double 
sum on the right-hand side of \eqref{eq:EntropyEnergySum} to energies inside the critical window $\mathcal{W}_L$.
Neglecting the possibility of cancellations between the two terms on the right-hand side of
\eqref{eq:Restdarstellung1} for the moment, 
we thus argue that $|\langle\psi_{E}^{\omega}, [H^{\omega}_L, 1_{A}(X)]
\psi^{\omega}_{E'}\rangle|^{2} \sim L^{-2}$ for eigenvalues $E, E' \in  \mathcal{W}_L$ and 
$\omega\in(\Omega_{L}(\alpha))^{c}$, see \eqref{eq:GleichVerteilung}. 
Moreover, the spacing of energies inside the critical window $\mathcal{W}_L$ is $\sim L^{-1}$ according to 
\eqref{eq:LevelSpacing} so that the remaining double sum is approximated by the double integral 
\begin{equation}
	\mathbb{E}[S_{E_{F}}(A,L)] \gtrsim 
	\int_{-L^{-\alpha-1/2}}^{-L^{-1}}\!\!\dd E\;
	\int_{L^{-1}}^{L^{-\alpha-1/2}}\!\!\dd E'\;
	\frac{1}{(E'-E)^2} \sim \frac{1-2\alpha}{2} \; \ln L.
\end{equation}
We note that this logarithmic divergence stems exclusively from the artificial $L$-dependence of the larger box 
$\Gamma_{L}$ to which the operator is restricted, rather than $A$. 
Also, we have neglected the possibility of cancellations from the two terms in \eqref{eq:Restdarstellung1} 
in the above crude argument. In fact, this is wrong for fixed bounded $A$, but can be justified 
for an $L$-dependent region $A = \Lambda_{L}^{\gamma,\delta}$, see \eqref{eq:ExactBox1} below, 
which then leads to
Theorem~\ref{thm:GeneralExpectationLogLowerBound}. 
(In Section~\ref{sec:ee-mod-ee} we will control the error which results from replacing $\Gamma_{L}$ by $\Z$ 
and arrive at Theorem~\ref{thm:main}.)

To prevent the above-mentioned cancellations we ensure that the two sine functions in
\eqref{eq:Restdarstellung1} are of opposite sign and bounded away from zero in absolute value for 
sufficiently many eigenfunctions with eigenvalues in the critical window $\mathcal{W}_L$, see Lemma~\ref{lem:sine-sign} 
and the following ones. This is achieved by 
the particular positioning of 
\begin{equation}\label{eq:ExactBox1}
	\Lambda_{L}^{\gamma,\delta} := [L_{1},L_{2}-1] \cap \mathbb{Z} \quad \text{with}\quad
	L_1:=-L+\lfloor\gamma L\rfloor, \;
	L_2:=-L+\lfloor (\gamma+\delta)L\rfloor
\end{equation}
inside $\Gamma_{L}$. Here, 
\begin{equation}
	\label{delta-gamma-cond}
 	0 < \delta \ll \gamma <1
\end{equation}
are small numbers, and $\lfloor\pmb\cdot\rfloor$
denotes the standard floor function or Gauss bracket. 
Indeed, integrating the formula \cite[Sect.\ 12.2]{MR1042095}, \cite[Lemma~2]{JitomirskayaSchSt} 
\begin{equation}
	\label{eq:WinkelgeschwindigkeitGamma}
	\frac{\d}{\d E}\theta_\ell^\omega(E)=\left(r^\omega_\ell(E)\right)^{-2}
	\sum_{x=-L}^{\ell-1}\left(\phi_E^\omega(x)\right)^2,
\end{equation}
$\ell \ge -L$, for the energy derivative of the Prüfer angles and using 
\eqref{eq:LevelSpacing} and
\eqref{eq:GleichVerteilung}, we conclude that the first sine function in \eqref{eq:Restdarstellung1} 
is sampled at a step size \mbox{$\sim\gamma + \delta$} when $E$ and $E'$ run through the eigenvalues while the 
second sine function is sampled at a step size \mbox{$\sim \gamma$}. 
Since $\delta \ll \gamma$, we will get sufficiently many good contributions.

\subsection{Finding good contributions}
\label{sec:good-cont}

Here we identify sufficiently many good contributions for a lower bound on the modulus of 
\eqref{eq:Restdarstellung1}.

The Fermi energy $E_{F} \in \{0,v\}$ coincides with one of the critical energies $E_{c}$ from Theorem~\ref{thm:Jitomirskaya}. Throughout this subsection $L \ge L_{\mathrm{min}}$ and $\alpha > 0$ from Theorem~\ref{thm:Jitomirskaya} will be fixed. 
We only consider elementary events $\omega\in(\Omega_{L}(\alpha))^{c}$ in this subsection. For the sake of 
readability we drop $\omega$ in the notation of all quantities. 
The enumeration 
\begin{equation}
 	E_{J_{\mathrm{min}}} <\cdots < E_{-2} <E_{-1} <
	E_{F} \le E_0< E_{1} < \hdots<E_{J_{\mathrm{max}}}
\end{equation}
of the $2L$ non-degenerate eigenvalues of $H_L$ will be convenient. 
The labelling index runs from the negative integer $J_{\mathrm{min}}$ to the positive 
integer $J_{\mathrm{max}}$, which both depend on $\omega$. 
Since we are only interested in eigenvalues in the critical window 
$\mathcal{W}_L$ below or above the Fermi energy, we introduce the two index sets
\begin{equation}
	\label{J-def}
	\begin{split}
		\mathcal{J}_{<} &:= \big\{ J_{\mathrm{min}} < j < 0 : E_{j-1}, E_{j} \in \mathcal{W}_L \big\},
		\\
		\mathcal{J}_{\ge} &:= \big\{ 0 \le j < J_{\mathrm{max}} : E_{j}, E_{j+1} \in \mathcal{W}_L  \big\}
	\end{split}
\end{equation}
and 
\begin{equation}
 	\label{J-grenze}
	J_{<} := \min\mathcal{J}_{<} -1, \qquad  J_{\ge} := \max\mathcal{J}_{\ge} +1.
\end{equation}
The next lemma analyses the step size at which the sine functions in \eqref{eq:Restdarstellung1} are sampled.

\begin{lemma}
	There exists a minimal length $L_{0} := L_{0}(C,\gamma,\delta) \in\N$ such that for every 
	length $L \ge L_{0}$ and every pair of consecutive eigenvalues $E_{j}, E_{j+1} \in \mathcal{W}_L$ we have
	\begin{equation}\label{eq:Winkelabstand}
		\begin{split}
			\theta_{L_1}(E_{j+1})-\theta_{L_1}(E_j) &\in \frac{\pi\gamma}{2C^6} \, [1, C^{12}],\\
			\theta_{L_2}(E_{j+1})-\theta_{L_2}(E_j)&\in \frac{\pi(\gamma+\delta)}{2C^6} \, [1,C^{12}].
		\end{split}
	\end{equation}
	Here, the quantity $C>1$ is the one from Theorem~\ref{thm:Jitomirskaya}, and $L_{1}, L_{2}$ are defined 
	in \eqref{eq:ExactBox1}.
\end{lemma}

\begin{proof}
	To prove the first statement in \eqref{eq:Winkelabstand} we use the explicit representation 
	(\ref{eq:WinkelgeschwindigkeitGamma}) for $(\d/\d E)\theta_{L_1}$. We see that for $E\in \mathcal{W}_L$ we get
\begin{equation}
\frac{\d}{\d E} \theta_{L_1}(E)=\sum_{n=-L}^{L_{1}-1}\left(\frac{\phi_E(n)}{r_{L_1}(E)}\right)^2=\frac{1}{2}\sum_{n=-L}^{L_1-1}\left(\frac{r_n(E)}{r_{L_1}(E)}\right)^2+\frac{1}{2}\left(\frac{\phi_E(L_1-1)}{r_{L_1}(E)}\right)^2.
\end{equation}
For $\omega\in(\Omega_L(\alpha))^c$ and all $n\in\Gamma_{L}$ we have, by the estimate \eqref{eq:GleichVerteilung} of Theorem \ref{thm:Jitomirskaya}, that
\begin{equation}
\left(\frac{r_{n}(E)}{r_{L_1}(E)}\right)^2\in\left[C^{-2},C^2\right].
\end{equation}
Hence,
\begin{equation}
\label{deriv-en}
\frac{\d}{\d E} \theta_{L_1}(E) \in \frac{\lfloor\gamma L\rfloor}{2}\left[C^{-2},C^2\right]+\frac{1}{2}\left[0,1\right]\subseteq  \frac{\gamma L}{2 C^3} \, [1, C^{6}],
\end{equation}
where the last inclusion holds for all $L \ge L_{0} := L_{0}(C,\gamma)$ because $C>1$. The first statement in \eqref{eq:Winkelabstand} now follows from integrating \eqref{deriv-en} over $E_{j+1}-E_{j}$ together with the estimate 
\eqref{eq:LevelSpacing}.

The verification of the second statement in \eqref{eq:Winkelabstand} is analogous and makes the minimal length $L_{0}$ also dependent on $\delta$.
\end{proof}


The commutator in \eqref{eq:Restdarstellung1} can be nicely bounded from below, if the two sine functions on the 
right-hand side are of opposite sign. This is done in 

\begin{lemma}
 	\label{lem:sine-sign}
	For $j\in \mathcal{J}_{\ge}$ and $k\in\mathcal{J}_{<}$ we define 
	\begin{equation}
 		z_{j,k}^{\pm} := \big\{[ \theta_{L_{2}}(E_j)-\theta_{L_{2}}(E_{k}) ] \pm [\theta_{L_1}(E_j)-\theta_{L_1}(E_k) ]\big\}/2.
	\end{equation}
	Assume that 
	\begin{equation}
		\label{z-cond}
 		|\cos z^+_{j,k}\sin z^-_{j,k} | \ge 1/2.
	\end{equation}
	Then 
	\begin{equation}
		|\langle\psi_{E_{k}}, [H_L,1_{\Lambda_{L}^{\gamma,\delta}}(X)]\psi_{E_{j}} \rangle| \geq \frac{1}{CL}
	\end{equation}
	holds with the quantity $C>1$ from Theorem~\ref{thm:Jitomirskaya}.
\end{lemma}

\begin{proof}
 	Introducing $\zeta^{\pm}_{j,k} :=  z_{j,k}^{+} \pm z_{j,k}^{-}$, the modulus of \eqref{eq:Restdarstellung1} reads
	\begin{equation}
		\label{neuer-Restdarstellung}
 		|\langle\psi_{E_{k}}, [H_L, 1_{\Lambda_{L}^{\gamma,\delta}}(X)]\psi_{E_{j}} \rangle| 
		= \big| r_{L_{2}}(E_{k})r_{L_{2}}(E_{j})\sin \zeta_{j,k}^{-} - r_{L_{1}}(E_{k})r_{L_{1}}(E_{j})\sin \zeta_{j,k}^{+} \big|.
	\end{equation}
	The condition \eqref{z-cond} implies that
	\begin{equation}
		\label{zeta-cond}
 		|\sin \zeta_{j,k}^{-} - \sin \zeta_{j,k}^{+} | = 2 \,|\cos z_{j,k}^{+} \sin z_{j,k}^{-} | \ge 1,
	\end{equation}
	and therefore that $\sin \zeta_{j,k}^{-}$ and $\sin \zeta_{j,k}^{+}$ have opposite signs. Thus, the right-hand side of 
	\eqref{neuer-Restdarstellung} equals
	\begin{equation}
 		 r_{L_{2}}(E_{k})r_{L_{2}}(E_{j}) |\sin \zeta_{j,k}^{-}| + r_{L_{1}}(E_{k})r_{L_{1}}(E_{j})|\sin \zeta_{j,k}^{+}|,
	\end{equation}
	and all four Pr\"ufer radii can be estimated from below with \eqref{eq:GleichVerteilung}. This yields the lower bound 
	\begin{equation}
 		|\langle\psi_{E_{k}}, [H_L,1_{\Lambda_{L}^{\gamma,\delta}}(X)]\psi_{E_{j}}\rangle| 
		\ge \frac{1}{CL} \, \big( |\sin \zeta_{j,k}^{-}| + |\sin \zeta_{j,k}^{+}| \big)
			= \frac{1}{CL} \, |\sin\zeta_{j,k}^{-} - \sin\zeta_{j,k}^{+}|. 
	\end{equation}
	Now, the claim follows from using again \eqref{zeta-cond}.
\end{proof}

From now on, our aim is to guarantee that condition \eqref{z-cond} is satisfied for sufficiently many indices $j$ and $k$.
We start with an auxiliary result.

\begin{lemma}
	\label{lem:z-diff}
	Let $j\in \mathcal{J}_{\ge}$ and $k\in\mathcal{J}_{<}$. Then
	\begin{equation}
		\label{eq:Intplusminus}
		\begin{split}
			z^-_{j+1,k} - z^-_{j,k} & \in \frac{\pi}{4 C^{6}} \; \big[ -(C^{12} -1) \gamma +\delta 
				\, , \,  (C^{12} -1) \gamma + C^{12} \delta \big] ,\\[.5ex]
			z^+_{j+1,k} - z^+_{j,k} & \in \frac{\pi(2 \gamma +\delta)}{4 C^{6}} \; [ 1, C^{12} ]
   \end{split}
	\end{equation}
	holds with the quantity $C>1$ from Theorem~\ref{thm:Jitomirskaya}.
\end{lemma}

\begin{proof}
	This statement is a direct consequence of \eqref{eq:Winkelabstand}, the identity	
	\begin{equation}
		z^\pm_{j+1,k}-z^\pm_{j,k} = \left\{ [\theta_{L_2}(E_{j+1})-\theta_{L_2}(E_j)] 
		\pm [\theta_{L_1}(E_{j+1})-\theta_{L_1}(E_j)] \right\}/2
	\end{equation}
	and that $[a,b]+[c,d]=[a+c,b+d]$ and $[a,b]-[c,d]=[a-d,b-c]$ holds for all finite intervals $[a,b],\;[c,d]\subseteq\R$.
\end{proof}

Now, we think of the index $k\in\mathcal{J}_{<}$ being fixed, whereas the index $j$ varies over
$\mathcal{J}_{\ge}$ in an increasing way in steps by one. For the time being, we assume the condition
\begin{align}\label{eq:Bedingung}
	C^{12}-1 <\frac{\delta}{\gamma}.
\end{align}
Its validity will be ensured later. Thus, according to \eqref{eq:Intplusminus}, 
both variables $z_{j,k}^{\pm}$ are strictly increasing functions in 
$j$ albeit $z_{j,k}^{+}$ grows much faster than $z_{j,k}^{-}$ due to \eqref{delta-gamma-cond}. Hence, the condition 
\eqref{z-cond} amounts to the requirement that sampling a beat produces an amplitude larger than $1/2$. First, we focus on the hull of the beat 
and divide it into antinodes. 

\begin{definition}
	\label{def:manysets}
	The set 
	\begin{equation}
 		\mathcal{Z}^{-} := \big\{ j \in \mathcal{J}_{\ge} : \, \sin z_{j,k}^{-} \sin z_{j-1,k}^{-} \le 0 
		\text{~and~} \sin z_{j,k}^{-}\neq 0 \big\}
	\end{equation}
	consists of those indices where a sign change occurs in the hull of the beat. It gives rise to a disjoint partition
	\begin{equation}
 		\mathcal{J}_{\ge} =: \bigcupdisjoint_{q=0}^{|\mathcal{Z}^{-}|} \mathcal{A}_{q}^{-} 
	\end{equation}
 	of the index set into ranges of successive indices forming an antinode of the hull. Here, we introduced
	$\mathcal{A}_{0}^{-} := \{ j\in\mathcal{J}_{\ge} : j < \min \mathcal{Z}^{-}\}$ as the left-most set in the partition (which is the only one that can be empty). 
	The requirement $\mathcal{A}_{q}^{-} 
	\cap \mathcal{Z}^{-}= \min \mathcal{A}_{q}^{-} $ for every $q\in \{1,\ldots,|\mathcal{Z}^{-}|\}$ renders the partition unique. 
	The set of ``hull-good'' indices in the $q$th antinode is defined as
	\begin{equation}
		\mathcal{J}_{q}^{\mathrm{hg}} := \big\{ j \in \mathcal{A}_{q}^{-} : |\sin z^-_{j,k}| \ge 2^{-1/2} \big\}, 
	\end{equation}
	and the set of ``good'' indices in the $q$th antinode as 
	\begin{equation}
		\mathcal{J}_{q}^{\mathrm{g}} := \big\{ j \in \mathcal{J}_{q}^{\mathrm{hg}} : |\cos z^+_{j,k}| \ge 2^{-1/2} \big\},
	\end{equation}
	 where $q \in \{1,\ldots,|\mathcal{Z}^{-}|-1\}$.	
 	For the sake of brevity, we have dropped the dependence on $k\in\mathcal{J}_{<}$ in all of the above notions. 
\end{definition}

\begin{remarks}
\item
	\label{rem:good-is-good}
	Clearly, $j \in \mathcal{J}_{q}^{\mathrm{g}}$ implies that \eqref{z-cond} holds for this index $j$ (and the fixed $k$).
\item
	We mention that $|\mathcal{Z}^{-}|$ grows with $L$ for large $L$:  according to \eqref{eq:LevelSpacing}, there are 
	$\mathcal O(L^{1/2-\alpha})$-many eigenvalues above $0$ in the critical window, and thus indices in $\mathcal{J}_{\ge}$. 
	The number of indices in an antinode of the hull is (large but) independent of $L$, see \eqref{eq:Intplusminus}, 
	\eqref{eq:Bedingung} and since $\delta\ll 1$.  
	Consequently, the number of antinodes is also of the order $\mathcal O (L^{1/2-\alpha})$. 
\end{remarks}

\begin{lemma}
	\label{prop1}
 	Fix $k\in\mathcal{J}_{<}$. We assume $C\le 2$, $\delta\le 2^{-7}$ and that \eqref{eq:Bedingung} holds. Then we have 
	\begin{equation}
		\label{num-antinodes}
		|\mathcal{A}_{q}^{-}| \in 2C^{6} \left[\frac{1}{(C^{12}-1)\gamma  + C^{12}\delta} \, , \,
			\frac{4}{- (C^{12} -1) \gamma  +\delta}\right]
	\end{equation}
	for every $q \in \{1,\ldots,|\mathcal{Z}^{-}|-1\}$. The upper bound in \eqref{num-antinodes} also holds for $q=0$ and $q=|\mathcal{Z}^{-}|$.
	Moreover, the number of hull-good indices is controlled by 
	\begin{equation}
		\label{num-hullgood}
		|\mathcal{J}_{q}^{\mathrm{hg}}| \in C^{6} \left[\frac{1}{(C^{12}-1)\gamma  + C^{12}\delta} \, , \,
			\frac{4}{- (C^{12} -1) \gamma  +\delta}\right]
	\end{equation}
	for every $q \in \{1,\ldots,|\mathcal{Z}^{-}|-1\}$. As before, $C>1$ stands for the quantity from Theorem~\ref{thm:Jitomirskaya}.
\end{lemma}

\begin{proof} 
	Lemma \ref{lem:z-diff} and \eqref{eq:Bedingung} provide the positive bounds 
	$a:= (\pi / 4 C^{6})[-(C^{12} -1) \gamma +\delta]$ and 
	$b := (\pi / 4 C^{6})[ (C^{12} -1) \gamma + C^{12}\delta]$ for the possible values of the increments 
	$(z_{j+1,k}^- - z_{j,k}^-) \in[a,b]$. For any $q\in\{1,\hdots,|\mathcal{Z}^{-}|-1\}$ the maximal phase difference of sample points within 
	an antinode can be estimated as 
	\begin{equation}
		\label{phase-diff-minus}
		\max_{j,l \in\mathcal A_q^-}\big\{z_{j,k}^- - z_{l,k}^-\big\}  \in\; ]\pi-2b , \pi[.
	\end{equation}
	Hence, we conclude
	\begin{equation}
		\label{}
 		|\mathcal{A}_{q}^{-}| \in \big[ \lfloor \pi/b\rfloor, \lfloor \pi/a\rfloor +1 \big] \subseteq
		[\pi/(2b), 2\pi/a].
	\end{equation}
	Here, the last inclusion holds, if $b\le \pi/2$ and $a \le \pi$. But $a\le b$ by definition, so the latter follows 
	from the former. 
	We point out that the assumptions 
	of the lemma even guarantee $b\le \pi/4$, which is needed below. This establishes \eqref{num-antinodes}. 
	Since the upper bound for the phase difference in \eqref{phase-diff-minus} is trivial and also holds for $q=0$ and 
	$q=|\mathcal{Z}^{-}|$, we infer the validity of the upper bound in \eqref{num-antinodes} for those two values of $q$, too. 

	Now, we turn to the proof of \eqref{num-hullgood}.
	Because of $\big|\{\varsigma\in [0,\pi] :|\sin\varsigma| \ge 2^{-1/2}\}\big|=\pi/2$, the maximal phase difference associated 
	with hull-good indices is restricted to
	\begin{equation}
		\max_{j,l \in\mathcal{J}^{\mathrm{hg}}_{q}}\big\{z_{j,k}^- - z_{l,k}^-\big\}  \in \; ](\pi/2)-2b , \pi/2].
	\end{equation}
	Similarly, we conclude
	\begin{equation}
		\label{hull-good-proof}
 		|\mathcal{J}^{\mathrm{hg}}_{q}| \in \big[ \lfloor \pi/(2b) -2\rfloor +2 , \lfloor \pi/(2a)\rfloor +1 \big] \subseteq
		[\pi/(4b), \pi/a],
	\end{equation}
	where the last inclusion follows from $a \le b\le \pi/4$.  
\end{proof}

Next, we assert that there are sufficiently many good indices per antinode.

\begin{lemma}
	\label{lem:good-ind-lb}
 	Fix $k\in\mathcal{J}_{<}$. We assume $C \le 2$, $\gamma \le 2^{-8}$, $\delta/\gamma \le 2^{-17}$ and that 
	\eqref{eq:Bedingung} is fulfilled. Then we have 
	\begin{equation}
		|\mathcal{J}_{q}^{\mathrm{g}}|  \ge  \frac{1}{2^{5} C^{18} \delta}
	\end{equation}
	for every $q \in \{1,\ldots,|\mathcal{Z}^{-}|-1\}$. Again, $C>1$ stands for the quantity from Theorem~\ref{thm:Jitomirskaya}.
\end{lemma}

\begin{proof} 
	The set
	\begin{equation}
 		\mathcal{Z}_{q}^{+} := \big\{ j \in \mathcal{J}^{\mathrm{hg}}_{q} : \, \cos z_{j,k}^{+} \cos z_{j-1,k}^{+} \le 0 
		\text{~and~} \cos z_{j,k}^{+}\neq 0 \big\},
	\end{equation}
	where $q \in \{1,\ldots,|\mathcal{Z}^{-}|-1\}$, consists of those indices where a sign change occurs in the fast 
	oscillation of the beat within the hull-good part of $q$th antinode. 
	It gives rise to a disjoint partition
	\begin{equation}
		\label{hull-good-decomp}
 		\mathcal{J}^{\mathrm{hg}}_{q} =: \bigcupdisjoint_{r=0}^{|\mathcal{Z}_{q}^{+}|} \mathcal{A}_{q,r}^{+} 
	\end{equation}
 	into ranges of successive indices forming antinodes of the fast oscillation. Here, we introduced
	$\mathcal{A}_{q,0}^{+} := \{ j\in \mathcal{J}^{\mathrm{hg}}_{q} : j < \min \mathcal{Z}_{q}^{+}\}$ as the left-most set in the partition (which is the only one that can be empty). The requirement $\mathcal{A}_{q,r}^{+} 
	\cap \mathcal{Z}_{q}^{+}= \min \mathcal{A}_{q,r}^{+} $ for every $r\in \{1,\ldots,|\mathcal{Z}_{q}^{+}|\}$ renders the partition unique. 
	
	First, we estimate the cardinality of $\mathcal{A}_{q,r}^{+}$ in the same way as it was done for  
	$\mathcal{A}_{q}^{-}$ in the proof of the previous lemma. 
	Lemma \ref{lem:z-diff} provides the positive bounds 
	$a' := \pi(2\gamma+\delta)/(4C^{6})$ and $b' := \pi C^{6}(2\gamma+\delta)/4$
	for the possible values of the increments 
	$(z_{j+1,k}^+ - z_{j,k}^+) \in[a',b']$. For any $q\in\{1,\hdots,|\mathcal{Z}^{-}|-1\}$ 
	and any $r\in\{1,\hdots,|\mathcal{Z}_{q}^{+}|-1\}$ the maximal phase difference of sample points within 
	an antinode of the fast oscillation can be estimated as 
	\begin{equation}
		\label{phase-diff-plus}
		\max_{j,l \in\mathcal A_{q,r}^+}\big\{z_{j,k}^+ - z_{l,k}^+\big\}  \in\; ]\pi-2b' , \pi[.
	\end{equation}
	Hence, we conclude
	\begin{equation}
		\label{A-plus-est}
 		|\mathcal{A}_{q,r}^{+}| \in \big[ \lfloor \pi/b'\rfloor, \lfloor \pi/a'\rfloor +1 \big] \subseteq
		[\pi/(2b'), 2\pi/a'],
	\end{equation}
	where the last inclusion follows from $0 < a' < b'\le \pi/2$. In fact, the assumptions of the lemma even guarantee 
	$b'\le \pi/4$, which we need below.
	Since the upper bound for the phase difference in \eqref{phase-diff-plus} is trivial and also holds for $r=0$ and 
	$r=|\mathcal{Z}_{q}^{+}|$, we infer the validity of the upper bound in \eqref{A-plus-est} for those two values of $r$, too.

	In order to estimate the cardinality of $\mathcal{Z}_{q}^{+}$ for $q\in\{1,\hdots,|\mathcal{Z}^{-}|-1\}$, we infer from
	\eqref{hull-good-decomp} and \eqref{A-plus-est} that
	\begin{equation}
 	|\mathcal{J}^{\mathrm{hg}}_{q}| \le (|\mathcal{Z}_{q}^{+}| + 1) \, \frac{2\pi}{a'}.
	\end{equation}
	The assumptions of the present lemma imply those of Lemma~\ref{prop1} which yields
	\begin{equation}
		|\mathcal{J}^{\mathrm{hg}}_{q}|\ge d :=\frac{1}{2C^{6}\delta}.  
	\end{equation}
	Thus, we arrive at
	\begin{equation}
		\label{Zplus-lb}
 		|\mathcal{Z}_{q}^{+}| - 1 \ge \frac{da'}{2\pi} -2  \ge \frac{da'}{4\pi},
	\end{equation}
	where the second inequality holds because the assumptions of the lemma imply $da' \ge 8\pi$.

	The set of ``good'' indices in the $r$th anti\-node of the fast oscillation is defined as
	\begin{equation}
		\mathcal{J}_{q,r}^{\mathrm{g}} := \big\{ j \in \mathcal{A}_{q,r}^{+} : |\cos z^{+}_{j,k}| \ge 2^{-1/2} \big\} 
	\end{equation}
	so that 
	\begin{equation}
		\label{good-decomp}
 		\mathcal{J}^{\mathrm{g}}_{q} \supseteq \bigcupdisjoint_{r=1}^{|\mathcal{Z}_{q}^{+}|-1} 
		\mathcal{J}_{q,r}^{\mathrm{g}}.
	\end{equation}
	For any $q\in\{1,\hdots,|\mathcal{Z}^{-}|-1\}$ and any $r\in\{1,\hdots,|\mathcal{Z}_{q}^{+}|-1\}$, 
	the maximal phase difference of good sample points within 
	an antinode of the fast oscillation can be estimated as 
	\begin{equation}
		\label{phase-diff-plus2}
		\max_{j,l \in\mathcal{J}_{q,r}^{\mathrm{g}}}\big\{z_{j,k}^+ - z_{l,k}^+\big\}  \in\; ](\pi/2)-2b' , \pi/2].
	\end{equation}
	Here, we used the second statement from Lemma~\ref{lem:z-diff}, $z_{j+1,k}^{+}- z_{j,k}^{+} \in [a',b']$ with
	$a' := \pi(2\gamma+\delta)/(4C^{6})$ and $b' := \pi C^{6}(2\gamma+\delta)/4$.
	Therefore, we conclude as in \eqref{hull-good-proof}
	\begin{equation}
		\label{Jgqr-lb}
 		|\mathcal{J}^{\mathrm{g}}_{q,r}| \in \big[ \lfloor \pi/(2b') -2\rfloor +2 , \lfloor \pi/(2a')\rfloor +1 \big] \subseteq
		[\pi/(4b'), \pi/a'],
	\end{equation}
	where the last inclusion follows from $0 < a' < b'\le \pi/4$. Combining \eqref{good-decomp}, \eqref{Zplus-lb} 
	and \eqref{Jgqr-lb}, we obtain $|\mathcal{J}^{\mathrm{g}}_{q}| \ge d a'/(2^{4}b')$, which proves the lemma. 
\end{proof}

%
\subsection{The logarithmic lower bound}
%

We assemble the results from the previous subsections and deduce a deterministic logarithmic lower bound 
for the quadratic analogue to the modified entanglement entropy from Definition~\ref{def:quadratic-ent}.


\begin{theorem}
	\label{prop:LogLowerBound}
	Let $v \in \; ]0,2[$ and $E_{F} \in \{0,v\}$.
	We fix $\alpha \in\; ]0, 1/6[\,$ and  $\gamma \in \; ]0,2^{-17}[\,$. 
	In addition, we assume that the quantity $C>1$ 
	from Theorem~\ref{thm:Jitomirskaya} satisfies 
	\begin{equation}
		\label{cond:Csmall}
 		C < 1+\gamma^2.
	\end{equation}
	Then, there exists a minimal length $L_{0}>0$ such that for all 
	$L \ge L_{0}$ 
	\begin{equation}
		\label{det-lb}
		Q^{\omega}_{E_{F}}(\Lambda_L^{\gamma,\gamma^{2}},L) 
		\ge 2^{-13}(1 - 6\alpha)\ln L
	\end{equation}
	for all disorder realisations $\omega\in(\Omega_L(\alpha))^{c}$.
\end{theorem}

We argue in the Appendix that the assumption \eqref{cond:Csmall} can always be satisfied by choosing the 
disorder strength $v$ sufficiently small. This and taking the expectation leads to 

\begin{corollary}
	\label{cor:E-LogLowerBound}
	We fix $\gamma \in \;]0,2^{-17}[$\,. There exists a maximal disorder strength $v_{0} \in \;]0,2[$ 
	such that for every $v\in\; ]0, v_{0}]$ and $E_{F} \in \{0,v\}$ there is a minimal length 
	$L_{0}^\prime>0$ such that for all $L \ge L_{0}^\prime$ 
	\begin{equation}
		\label{ExpectationLogLowerBound}
		\mathbb{E}\big[  Q_{E_{F}}(\Lambda_L^{\gamma,\gamma^{2}},L)  \big] 		
		\ge 2^{-15}\ln L.
	\end{equation}
\end{corollary}

\begin{proof}
	As $\lim_{v\downarrow0} C =1$ by Theorem~\ref{thm:Jitomirskaya}, there exists a maximal disorder strength
	$v_{0}\in \;]0,2[$ such that \eqref{cond:Csmall} holds for every $v\in \;]0,v_{0}]$.
	We choose $\alpha=1/12$ in Theorem~\ref{prop:LogLowerBound} and infer from \eqref{det-lb} that 
	\begin{equation}
		\mathbb{E}\big[  Q_{E_{F}}(\Lambda_L^{\gamma,\gamma^{2}},L)  \big] 
		\ge   2^{-14}\ln L \; \mathbb{P}\big[(\Omega_{L}(\alpha))^{c}\big] 
	\end{equation}
	for every $L\ge L_{0}$. Now, the claim follows from \eqref{prob-bad-event}, possibly by enlarging $L_{0}$.
\end{proof}

\begin{proof}[Proof of Theorem \ref{prop:LogLowerBound}]
	Let $\omega\in (\Omega_L(\alpha))^{c}$ and, for the time being, $L \ge L_{\mathrm{min}}$. 
	We use the notation introduced at the beginning 
	of Section~\ref{sec:good-cont} and drop $\omega$ from all quantities, as it is also done there. 
	By restricting the double sum in  
	Lemma~\ref{prop:Summe} to energies inside the critical window, we arrive at the estimate
	\begin{equation}
		\label{S-j<>}
		Q_{E_{F}}(\Lambda_L^{\gamma,\gamma^{2}},L) 
		\ge 4 \sum_{j\in\mathcal{J}_{\ge}, \, k \in\mathcal{J}_{<}} \,
			\frac{1}{(E_{j}-E_{k})^{2}} \; \Big|\langle\psi_{E_{k}}, [H_L,1_{\Lambda_{L}^{\gamma,\gamma^{2}}}(X)]
			\psi_{E_{j}}\rangle \Big|^2.
	\end{equation}
	We aim to apply the lower bound for the commutator from Lemma~\ref{lem:sine-sign}. Its assumption \eqref{z-cond} is 
	satisfied for every fixed $k\in \mathcal{J}_{<}$ after further restricting the $j$-sum to good indices according to 
	$\mathcal{J}_{\ge} \supseteq \bigcup_{q=1}^{|\mathcal{Z}^{-}|-1} \mathcal{J}_{q}^{\mathrm{g}}$, see 
	Remark~\ref{rem:good-is-good}. This gives the lower bound
	\begin{equation}
		\label{S-kq}
		\frac{4}{(CL)^{2}} \sum_{k \in\mathcal{J}_{<}} \sum_{q=1}^{|\mathcal{Z}^{-}|-1}
		\sum_{j\in \mathcal{J}_{q}^{\mathrm{g}}}  \, \frac{1}{(E_{j}-E_{k})^{2}} 
		\ge \frac{4}{(CL)^{2}} \sum_{k \in\mathcal{J}_{<}} \sum_{q=1}^{|\mathcal{Z}^{-}|-1} \, 
			\frac{|\mathcal{J}_{q}^{\mathrm{g}}|}{(\varepsilon_{q}^{(k)} - E_{k})^{2}}
	\end{equation}
	for the right-hand side of \eqref{S-j<>}, where we introduced 
	\begin{equation}
 		\varepsilon_{q}^{(k)} := \max_{j\in \mathcal{A}_{q}^{-}} E_{j}
	\end{equation}
	for $q= 1,\ldots,|\mathcal{Z}^{-}|$ and $k \in\mathcal{J}_{<}$. We recall that there is a suppressed $k$-dependence in 
	the quantities of Definition~\ref{def:manysets} which we made explicit again in $\varepsilon_{q}^{(k)}$.

	The assumptions of the theorem imply those of Lemma~\ref{lem:good-ind-lb} because the elementary inequality 
	$(1+\rho)^{n} \le 1+ 2^{n}\rho$, valid for $\rho \in [0,1]$ and $n\in\mathbb{N}$, ensures that \eqref{eq:Bedingung} holds. 
	In fact, even the stronger inequality 
	\begin{equation}
		\label{eq:Bedingung3}
 		C^{12} - 1 \le 2 ^{12} \gamma^{2}  \le 2^{-5} \gamma = 2^{-5} \delta/\gamma
	\end{equation}
	is fulfilled.
	Therefore, we can apply the lemma and infer that the expression 
	\begin{equation}
		\label{S-pre-int}
		\frac{1/L}{2^{3}C^{20}\delta} \sum_{k \in\mathcal{J}_{<}} \sum_{q=1}^{|\mathcal{Z}^{-}|-1} \, 
		\frac{\varepsilon_{q+1}^{(k)} - \varepsilon_{q}^{(k)}}{(\varepsilon_{q}^{(k)} - E_{k})^{2}} \;
		\frac{1/L}{\varepsilon_{q+1}^{(k)} - \varepsilon_{q}^{(k)}}
	\end{equation}
	is a lower bound for the right-hand side of \eqref{S-kq}. The energy $\varepsilon_{q}^{(k)}$ 
	is the rightmost in the 
	$q$th antinode of the hull. Therefore we can estimate their differences as
	\begin{equation}
		\label{increment-lb}
 		0 < \varepsilon_{q+1}^{(k)} - \varepsilon_{q}^{(k)} \le |\mathcal{A}_{q+1}^{-}|  \,
		\frac{\pi C^{3}}{L} \le \frac{\pi 2^{3} C^{9}}{L\delta} \, \frac{1}{1-2^{-5}}  
		\le \frac{2^{6} C^{9}}{L\delta} 
	\end{equation}
	for $q= 1,\ldots,|\mathcal{Z}^{-}|-1$, independently of $k$.
	Here, the first upper bound on the difference follows from \eqref{eq:LevelSpacing} and the second from Lemma~\ref{prop1} and
	\eqref{eq:Bedingung3}. We note that the assumptions of Lemma~\ref{prop1} are weaker than those of 
	Lemma~\ref{lem:good-ind-lb}.
	Combining \eqref{S-j<>}, \eqref{S-kq}, \eqref{S-pre-int} and \eqref{increment-lb}, we arrive at 
	\begin{equation}
		\label{S-first-int}
 		Q_{E_{F}}(\Lambda_L^{\gamma,\gamma^{2}} \!\!,L)  
		\ge   \frac{1/L}{2^{9}C^{29}} \sum_{k \in\mathcal{J}_{<}} \sum_{q=1}^{|\mathcal{Z}^{-}|-1}  
		\frac{\varepsilon_{q+1}^{(k)} - \varepsilon_{q}^{(k)}}{(\varepsilon_{q}^{(k)} - E_{k})^{2}} 
		\ge \frac{1/L}{2^{9}C^{29}}\, \sum_{k \in\mathcal{J}_{<}} 
				\int_{\varepsilon_{1}^{(k)}}^{\varepsilon^{(k)}_{|\mathcal{Z}^{-}|}} \frac{\d\varepsilon}{(\varepsilon - E_{k})^{2}}.
	\end{equation}
	The goal is now to deduce a $k$-independent lower bound on the range of the $\varepsilon$-integration. 
	This will allow to interchange the integral with the $k$-sum. Since $\varepsilon_{1}^{(k)}$, respectively  
	$\varepsilon_{|\mathcal{Z}^{-}|}^{(k)}$, lies in the first ($q=1$), respectively last,  
	antinode of the hull, we estimate as in \eqref{increment-lb}
	\begin{equation}
		\label{eps-int-bound}
		\begin{split}
 			\varepsilon_{1}^{(k)} & \le E_{F} +  \big( |\mathcal{A}_{0}^{-}| + |\mathcal{A}_{1}^{-}| \big) \,	\frac{\pi C^{3}}{L} 
				\le E_{F} + \frac{2^{7} C^{9}}{L\delta}, \\
			\varepsilon_{|\mathcal{Z}^{-}|}^{(k)} & \ge \max \mathcal{W}_L -  \frac{\pi C^{3}}{L}
				\ge E_{F} + L^{-1/2-\alpha} -  \frac{2^{2} C^{3}}{L},
		\end{split}
	\end{equation}
	independently of $k$.	
	Let $L_{0} \ge L_{\mathrm{min}}$ be so large that both $2^{2}C^{3} \le 2^{7}C^{9}/\delta \le 2^{16} /\delta \le L_{0}^{\alpha}$
	and $L_{0}^{\alpha} - L_{0}^{-1/2+3\alpha} \ge 1$. 
	For the rest of this proof we assume $L \ge L_{0}$. Then \eqref{eps-int-bound} simplifies to
	\begin{equation}
		\label{eps-int-bound2}
		\begin{split}
 			\varepsilon_{1}^{(k)} & \le E_{F} + L^{-1+\alpha}, \\
			\varepsilon_{|\mathcal{Z}^{-}|}^{(k)} &	\ge E_{F} + L^{-1/2-2\alpha} (L^{\alpha} - L^{-1/2+3\alpha} ) \ge E_{F} + L^{-1/2-2\alpha}.
		\end{split}
	\end{equation}
	We recall the definition of $J_{<}$ from \eqref{J-grenze} and conclude that 
	\begin{align}
		\label{S-two-changed}
	 	Q_{E_{F}}(\Lambda_L^{\gamma,\gamma^{2}},L) 	
		&\ge \frac{1}{2^{9}C^{29}}\, \int_{E_{F}+ L^{-1 +\alpha}}^{E_{F}+ L^{-1/2-2\alpha}} 
			\!\d\varepsilon  \sum_{k \in\mathcal{J}_{<}} 
			\frac{E_{k} -E_{k-1}}{(\varepsilon - E_{k})^{2}} \, \frac{1/L}{E_{k} -E_{k-1}} \notag\\
			& \ge \frac{1}{2^{11}C^{32}}\, \int_{E_{F}+ L^{-1 +\alpha}}^{E_{F}+ L^{-1/2-2\alpha}} \!
			\d\varepsilon 	\int_{E_{J_{<}}}^{E_{-1}} \!\d\eta \, \frac{1}{(\varepsilon-\eta)^{2}}	
	\end{align}
	because the level-spacing estimate \eqref{eq:LevelSpacing} provides the bound $E_{k}-E_{k-1} \le \pi C^{3}/L$.
	It also implies
	\begin{equation}
		\begin{split}
 			E_{-1} & \ge E_{F} - 	\frac{\pi C^{3}}{L} \ge E_{F} - L^{-1+\alpha}, \\
			E_{J_{<}} & \le  \min \mathcal{W}_L + \frac{\pi C^{3}}{L} \le E_{F} - L^{-1/2-2\alpha}, 
		\end{split}
	\end{equation}
	where we argued similarly as in \eqref{eps-int-bound2} for $L\ge L_{0}$. We thus estimate and integrate
	\begin{equation}
		\label{S-two-int}
	 	Q_{E_{F}}(\Lambda_L^{\gamma,\gamma^{2}},L)  \ge \frac{1}{2^{11}C^{32}}\, \int_{L^{-1 +\alpha}}^{L^{-1/2-2\alpha}} \!\d\varepsilon 	
			\int_{-L^{-1/2-2\alpha}}^{-L^{-1+\alpha}} \!\d\eta \, \frac{1}{(\varepsilon-\eta)^{2}}
		\ge \frac{1-6\alpha}{2^{12}C^{32}}\, \ln L.		
	\end{equation}
 	Finally, the estimate $C^{32} \le 1 + 2^{32}\gamma^{2} \le 2$, which follows from the elementary inequality 
	above \eqref{eq:Bedingung3}, yields the claim.
\end{proof}

%

To conclude this section, we sketch the necessary modifications for the proof of Theorem~\ref{thm:GeneralExpectationLogLowerBound}. 
The goal is to obtain a similar statement to Theorem \ref{prop:LogLowerBound}, which is valid for all possible coupling constants $v\in\; ]0,2[$. We therefore cannot rely on $C$ being arbitrarily close to one.

\begin{proof}[Proof of Theorem \ref{thm:GeneralExpectationLogLowerBound}]
	\label{proof:boundary-box}
	The use of $\Lambda_{L}' :=[-L, -(1- \delta) L] \cap \Z = \Lambda_{L}^{0,\delta}$ amounts to 
	$\gamma=0$ in our previous arguments. This change simplifies the matrix elements of the 
	commutator \eqref{neuer-Restdarstellung} dramatically, since $\theta_{L_1}(E)=0$  in this case by definition 
	for all values $E\in\R$. Hence, 
	\begin{equation}
	|\langle\psi_{E_k},[H_L,1_{\Lambda_L'}(X)]\psi_{E_j}\rangle|=\big|r_{L_2}(E_k)r_{L_2}(E_j)\sin (2z^-_{j,k})\big|\ge \frac{1}{CL}|\sin (2z^-_{j,k})|
	\end{equation}
	for all $k\in\mathcal J_{<}$ and $j\in\mathcal J_\ge$. This renders the considerations of 
	Lemma~\ref{lem:sine-sign} and Lemma~\ref{lem:good-ind-lb} unnecessary, because only a single 
	sine-function shows up. The overall argument, however, is similar to the one in Lemma~\ref{prop1} 
	with an additional factor of $2$. 

	We therefore redefine the set 
	\begin{equation}
 		\mathcal{Z}^{-} := \big\{ j \in \mathcal{J}_{\ge} : \, \sin (2z_{j,k}^{-}) \sin (2z_{j-1,k}^{-}) \le 0 
		\text{~and~} \sin (2z_{j,k}^{-})\neq 0 \big\}
	\end{equation}
	of indices where a sign change occurs in the oscillation. As before, this gives rise to a disjoint partition
	\begin{equation}
 		\mathcal{J}_{\ge} =: \bigcupdisjoint_{q=0}^{|\mathcal{Z}^{-}|} \mathcal{A}_{q}^{-} 
	\end{equation}
 	of the index set into ranges of successive indices forming an antinode of the oscillation. 
	The set of good indices in the $q$th antinode is defined as
	\begin{equation}
		\mathcal{J}_{q}^{\mathrm{g}} := \big\{ j \in \mathcal{A}_{q}^{-} : |\sin (2z^-_{j,k})| \ge 2^{-1/2} \big\}.
	\end{equation}
	The proof of Lemma~\ref{lem:z-diff} is valid for $\gamma=0$. It provides positive bounds 
	$a:=\delta\pi/2C^6$ and $b:=\delta\pi C^6/2$ for the positive values of the increments 
	$2(z^-_{j+1,k}-z^-_{j,k})\in[a,b]$. From the proof of Lemma~\ref{prop1} we get the following estimate
	\begin{equation}\label{thm1.2-LowerBoundGood}
		|\mathcal J^{\textrm{g}}_{q}|\ge \lfloor \pi/(2b)\rfloor 
		\ge \pi/(4b),
	\end{equation}
	where the last inclusion follows from $b\le \pi/4$, which is true for $\delta<2C^{-6}$. The 
	inequality \eqref{thm1.2-LowerBoundGood} replaces the estimate of Lemma~\ref{lem:good-ind-lb}. 
	The rest of the proof is identical to the one of Theorem~\ref{prop:LogLowerBound} and 
	Corollary~\ref{cor:E-LogLowerBound}, except that we do not take the expectation at the end 
	but appeal to the Borel--Cantelli Lemma to conclude that 
	\begin{equation}
 		\mathbb{P}\big\{\omega\in \big(\Omega_{L}(\alpha)\big)^{c} \text{ for finally all } 
			L\in\N\big\} =1.
	\end{equation}
%
%
%
%
\end{proof}


\section{Proof of Theorem \ref{thm:main}}
\label{sec:ee-mod-ee}

In this section we deduce the main Theorem \ref{thm:main} from Corollary \ref{cor:E-LogLowerBound}. 
The goal is to control the error arising from the replacement of the outer volume $\Gamma_L$ in the 
modified entanglement entropy by the whole space $\Z$. We consider a discrete interval 
$A\subset\Gamma_L$ and denote by $f(H^\omega_L)$ the trivial 
extension of this operator from $\ell^2(\Gamma_{L})$ to the space $\ell^2(\Z)$ 
for any measurable function $f:\,\R\rightarrow\R$.

Our strategy is to apply Kre\u\ii n's trace formula, see e.g.\ \cite[Sect.\ 9.7]{schmuedgen2012unbounded},
\begin{multline}
	\label{eq:Kreinsformula}
	\left|\tr{g\big(1_{A}(X) 1_{<E_F}(H^\omega_L)1_{A}(X)\big) -
		g\big(1_{A}(X)1_{<E_F}(H^\omega)1_{A}(X)\big)}\right| \\
	=\Big|\int_0^1\dd s\;g^\prime(s)\, \xi_L(s)\Big|\le4\left\|\xi_L\right\|_{L^1}
\end{multline}
to the parabola $g$ from \eqref{hdef}, where  
\begin{equation}
  \xi_L:\;\R\ni s\mapsto \tr{1_{\le s}\big(1_{A}(X)1_{<E_F}(H^\omega)1_{A}(X)\big)
 	-1_{\le s}\big(1_{A}(X)1_{<E_F}(H_L^\omega)1_{A}(X)\big)}
\end{equation}
is the spectral shift function. 
Here, $\|\cdot\|_{L^1}$ denotes the $L^1(\R)$-norm.
It can be estimated in terms of the trace norm $\|\cdot\|_1$ of the difference 
\begin{equation}\label{eq:Krein}\|\xi_L\|_{L^1}\le\big\|1_{A}(X)\big(1_{<E_F}(H^\omega)-1_{<E_F}(H^\omega_L)\big)1_{A}(X)\big\|_1.
\end{equation}


We recall from Theorem~\ref{thm:Jitomirskaya} that the density of states $\mathcal{N}'$ exists at the critical energies 
$E_{F}= 0$ and $v$. For $L\in\N$ we define $d_L:=\textrm{dist}(A,\{-L,L-1\})$ as the distance of the 
small box to the boundary of the big box.

\begin{lemma}
	\label{shift-function}
	Let $E_{F} \in\{0,v\}$ and $\alpha >0$. Then there exists a minimal length
	$L_{0} \in\N$, which depends only on $\alpha$ and on the model parameters, such that for all 
	$L \ge L_{0}$, all ``temperatures'' $T\in\;]0,\infty[$ and all discrete intervals $A\subset\Gamma_{L}$
	we have the estimate
	\begin{equation}
		\label{eq:EstimateShiftFunction}
		\E{\left|Q_{E_{F}}(A,L)  - Q_{E_{F}}(A) \right| }  \le M(A,T)  + R(A,T)
	\end{equation}
	with a main term
	\begin{equation}
 		M(A,T) := 2^{4}C^{4}|A|T
	\end{equation}
	and a remainder term
	\begin{equation}
		R(A,T) := 2^{5}C +  2^7|A|^2 \big(T^{-2} \e^{-d_LT/6} + \;T \e^{-d_L/6} \big) 
			+ 2^{5}C^{3}|A| \e^{-L^{-1/2-\alpha}/T}.
	\end{equation}
\end{lemma}

The lemma is proven in Subsection~\ref{se:lem-sf} below. Before we turn to the proof of the main theorem, 
we need another perturbation result of a similar spirit.

\begin{lemma}
	\label{lem:AAprime}
	Let $A,A' \subseteq \Gamma_{L}$ and $E_{F}\in\R$. Then 
	\begin{equation}
		\label{eq:AAprime}
 		\left| Q_{E_{F}}(A,L) - Q_{E_{F}}(A',L) \right| \le 4 r,
	\end{equation}
 	where $r$ denotes the cardinality of the symmetric difference of $A$ and $A'$.
\end{lemma}

\begin{proof}
	We use the abbreviation $P:=1_{E_{F}}(H_{L})$.
	The operators $1_{A^{(\prime)}}(X)P 1_{A^{(\prime)}}(X)$ and $P1_{A^{(\prime)}}(X)P$ share the same non-zero singular values. This and $g(0)=0$ implies that	the left-hand side of \eqref{eq:AAprime} equals 
	\begin{align}
 		\left|\Tr \big\{ g\big(P 1_{A}(X)P\big) - g\big(P1_{A'}(X)P\big) \big\}\right| 
		& \le 4 \|P \big(1_{A}(X) - 1_{A'}(X)\big)P \|_{1}  \notag\\
		& \le 4  \|1_{A}(X) - 1_{A'}(X) \|_{1} = 4 r,
	\end{align}
	where, in order to deduce the first inequality, we argued with Kre\u\ii n's trace formula as in \eqref{eq:Kreinsformula} and \eqref{eq:Krein} but with the operators $P 1_{A}(X)P$ and $P 1_{A'}(X)P$.
\end{proof}

\begin{proof}[Proof of Theorem \ref{thm:main}]
	We fix $\alpha:=1/12$ and $\gamma\in\; ]0,2^{-26}[$. The goal is to apply Lemma~\ref{shift-function}
	with $A=\Lambda_{L}^{\gamma,\gamma^{2}}$, whence $d_L=\lfloor\gamma L\rfloor$ and 
	$|\Lambda_L^{\gamma,\gamma^{2}}|=\lfloor(\gamma+\gamma^2)L\rfloor-\lfloor \gamma L\rfloor$. 

	First, we have to replace the box $\Lambda_{L}=\{1,\ldots,L\}$ by the differently positioned box 
	$\Lambda_{L}^{\gamma,\gamma^{2}}$ according to
	\begin{equation}
		\label{eq:diffLambda}
		\liminf_{L\rightarrow\infty}\frac{\mathbb{E}\big[S_{E_{F}}(\Lambda_L)\big]}{\ln L}
		\ge \liminf_{L\rightarrow\infty}\frac{\mathbb{E}\big[Q_{E_{F}}(\Lambda_L) 
			\big]}{\ln L} 	
		= \liminf_{L\rightarrow\infty} \frac{\mathbb{E}\big[ Q_{E_{F}}(\Lambda_{L}^{\gamma,\gamma^{2}})
			\big] }{\ln L} .  
	\end{equation}	
	As to the validity of the equality in \eqref{eq:diffLambda} we remark that 
	~(i)~ $\lim_{L\to\infty} \ln L/\ln |\Lambda_L^{\gamma,\gamma^{2}}| =1$, ~(ii)~
	$ \exists\, \tilde{L}\in \N$ such $\{ |\Lambda_L^{\gamma,\gamma^{2}}| : L\ge\tilde L \} = \N$, 
	see Lemma~\ref{lem:lengthcount},
	~(iii)~ ergodicity with respect to $2\Z$-translations allows to shift $\Lambda_L^{\gamma,\gamma^{2}}$ 
	such that its left-most point is either $0$ or $1$ and ~(iv)~ if it is $0$, then Lemma~\ref{lem:AAprime} 
	allows us to shift it to $1$ at the cost of an $L$-independent error of the numerator not larger than $8$ 
	so that this error
	does not contribute in the limit $L\to\infty$. 
	
	Introducing the abbreviation $\mathcal{E}_{L} := \mathbb{E}\big[ |Q_{E_{F}}(\Lambda_{L}^{\gamma,
		\gamma^{2}},L) - Q_{E_{F}}(\Lambda_{L}^{\gamma,\gamma^{2}})|\big]$, \eqref{eq:diffLambda} implies
	\begin{align}
		\label{eq:final-est-ent}
		\liminf_{L\rightarrow\infty}\frac{\mathbb{E}\big[S_{E_{F}}(\Lambda_L)\big]}{\ln L}
		&	\ge \liminf_{L\rightarrow\infty} \frac{\mathbb{E}\big[ Q_{E_{F}}(\Lambda_{L}^{\gamma,\gamma^{2}},L)
			\big] }{\ln L} 
			- \limsup_{L\rightarrow\infty} \frac{\mathcal{E}_{L}}{\ln L} \notag\\
		& \ge 2^{-15}	 - \limsup_{L\rightarrow\infty} \frac{\mathcal{E}_{L}}{\ln L},
	\end{align}	
	where we  used Corollary~\ref{cor:E-LogLowerBound} in the last step, assuming that $v \in \;]0, v_{0}]$.

	Now, we estimate the error $\mathcal{E}_{L}$ with Lemma~\ref{shift-function}. To do so, we choose 
	the temperature as	$T_L:=(K \ln L)/L$ with some constant 
	\begin{equation}
		\label{ersteBedanK}
 		K>24/\gamma.
	\end{equation}
	We find for the remainder term that
	\begin{equation}
 		\lim_{L\to\infty} R\big(\Lambda_{L}^{\gamma,\gamma^{2}},T_{L}\big) = 2^{5}C,
	\end{equation}
	where we used \eqref{ersteBedanK} to see that the contribution proportional to 
	$|\Lambda_{L}^{\gamma,\gamma^{2}}|^{2}T_{L}^{-2}\e^{-d_{L}T_{L}/6} \sim [L^{4}/(\ln L)^{2}]
		\e^{-\gamma K \ln L /6}$ vanishes in the limit. Thus, we deduce from \eqref{eq:final-est-ent} that 
	\begin{equation}
		\label{eq:main-term-est}
		\liminf_{L\rightarrow\infty}\frac{\mathbb{E}\big[S_{E_{F}}(\Lambda_L)\big]}{\ln L}
		\ge 2^{-15}	 - \mathcal{L}_{v}
	\end{equation}	
	with 
	\begin{equation}
 		\mathcal{L}_{v} := \limsup_{L\rightarrow\infty} \frac{M\big(\Lambda_{L}^{\gamma,\gamma^{2}},T_{L}\big)}{\ln L}
		= 2^{4}C^{4}\gamma^{2} K
	\end{equation}
	being even a limit.
	
	The claim of the theorem then follows from \eqref{eq:main-term-est} and the requirement 
	\begin{equation}
		\label{calL-cond}
 		\mathcal{L}_{v} \le 2^{-16}  \quad\text{for all } v \in \;]0,v_{0}].
	\end{equation} 
	To see the validity of \eqref{calL-cond} we recall from the proof of  
	Corollary~\ref{cor:E-LogLowerBound} that the restriction $v\le v_{0}$ guarantees the bound $C < 1+ \gamma^{2}$. Since $\gamma \in \;]0, 2^{-26}[$, we have $C^{4} \le 2$. Therefore Inequality \eqref{calL-cond} follows if
	\begin{equation}
		\label{zweiteBedanK}
		2^6\gamma^2 K\le 2^{-16}.
	\end{equation}
 	But, since $\gamma \le 2^{-26}$, the two conditions  \eqref{ersteBedanK} and \eqref{zweiteBedanK} do 
	not contradict each other, and such a constant $K$ does indeed exist. 
\end{proof}

\subsection{Proof of Lemma~\ref{shift-function}}
\label{se:lem-sf}

Without loss of generality we restrict ourselves in the proof of Lemma~\ref{shift-function} to the case 
$E_{F}=0$, the other case being analogous. 

According to \eqref{eq:Kreinsformula} -- \eqref{eq:Krein} we will estimate the trace norm of the 
difference $1_{A}(X)(1_{<0}(H^\omega)-1_{<0}(H^\omega_L))1_{A}(X)$. To do so we write the 
Fermi projections as contour integrals over the resolvent. Then the well-known geometric resolvent equation
and the Combes--Thomas estimate will allow us to estimate the integrand. The first step, however, requires 
an analytical function of the Schr\"odinger operator. Hence, we replace the Fermi projection $1_{<0}$
by the analytical Fermi--Dirac distribution $f_{T}:=1/(1+ \e^{(\,\pmb{\cdot}\,)/T})$ with temperature $T>0$.

\begin{lemma}
	\label{lem:CombesThomas}
	The deterministic estimate 
	\begin{align}
		\label{eq:GlaettungAbstand} 
		\left\|1_{A}(X)\left(f_{T}(H^\omega_L)-f_{T}(H^\omega)\right)1_{A}(X)\right\|_1
		\le 2^5 |A|^2 \left( T \e^{-d_L/6} + T^{-2} \e^{-d_L T/6}\right)
	\end{align}
	holds for all $L\in\N$, $T>0$ and $\omega\in\Omega$.
\end{lemma}

\begin{proof}
The function $f_{T}$ is  holomorphic on the strip  $\{z\in\C:|\Im(z)|<\pi T\}$. Let $\gamma_T$ be a curve 
encircling $\sigma(H^\omega_L)$ and $\sigma(H^\omega)$ counter-clockwise in this strip for all $L\in\N$
and all $\omega\in\Omega$. We choose $\gamma_T$ such that its image borders the rectangle 
\begin{equation}
 	\big\{z\in\C: |\Im(z)|\le \min(1,T\pi/2),\; \Re(z)\in[-3,5]\big\}
\end{equation}
and note that $\sigma(H_{(L)}^\omega)\subseteq[-2,4]$ for all $v\in\;]0,2[$, all $L\in\N$ and all 
$\omega\in\Omega$.
We conclude that 
\begin{multline}
1_{A}(X)\big(f_{T}(H^\omega)-f_{T}(H^\omega_L)\big)1_{A}(X) \\
=\frac{1}{2\pi \i}\oint_{\gamma_{T}}\dd z\; f_{T}(z) \, 1_{A}(X)\Big(\frac{1}{z-H^\omega}-\frac{1}{z-H^\omega_L}\Big)1_{A}(X).
\end{multline}
The geometric resolvent equation yields
\begin{multline}
	1_{A}(X)\Big(\frac{1}{z-H^\omega}-\frac{1}{z-H^\omega_L}\Big)1_{A}(X) \\
	=	-1_{A}(X)\frac{1}{z-H^\omega} \; \big(\proj{\delta_{-L-1}}{\delta_{-L}}+\proj{\delta_{L}}{\delta_{L-1}}\big) \; \frac{1}{z-H^\omega_L}
		1_{A}(X).
\end{multline}
We estimate the matrix elements of the resolvent with the Combes--Thomas estimate 
\cite[Thm.~11.2]{artRSO2008Kir} 
\begin{align}
	\big|\big\langle\delta_{x}, \Big(\frac{1}{z-H^\omega} - \frac{1}{z-H^\omega_L}\Big) \delta_{y} 
		\big\rangle\big|
	\le 2 \;
	\frac{2^2}{\textrm{dist}(z,[-2,4])^2} \; \e^{-2\, \textrm{dist}(z,[-2,4]) \, d_L/12}
\end{align}
for every $x,y\in A$ and every $z\notin\sigma(H_L^\omega)\cup\sigma(H^\omega)$.
As to the applicability of \cite[Thm.~11.2]{artRSO2008Kir}, we note that by inspection of the proof one obtains the statement 
not only for $z\in\C$ with distance to the spectrum $\le 1$, but even if it is $\le 12$, which is fulfilled in our case.

An elementary computation shows that $|f_{T}(z)|\le 1$  for all $z$ on the curve $\gamma_T$. 
Furthermore, for all $z$ on the horizontal parts of $\gamma_{T}$ where $|\Im(z)|= \min(1,T\pi/2)$ we find 
$\textrm{dist}(z,[-2,4])\ge T\pi/2$. On the vertical parts where $\Re(z)\in\{-3,5\}$ we have 
$\textrm{dist}(z,[-2,4])\ge 1$. Hence,

\begin{align}
	\big\|1_{A}(X)\big(f_{T}(H^\omega_L)  &-f_{T}(H^\omega)\big)1_{A}(X)\big\|_1  \notag \\
	& \le \sum_{x,y\in A}\frac{1}{2\pi}\left|\oint_{\gamma_T}\dd z \, f_{T}(z) \, 
		\big\langle\delta_{x}, \Big(\frac{1}{z-H^\omega} - \frac{1}{z-H^\omega_L}\Big) \delta_{y} 
		\big\rangle \right| \notag\\
	&\le 2^5 |A|^2 \left( T \e^{-d_L/6} + T^{-2} \e^{-d_L T/6}\right).
\end{align}
\end{proof}

Since we approximate the Fermi projection by $f_T(H^\omega_{(L)})$ we have to control an error term, 
which we estimate in the following two lemmata.

\begin{lemma}\label{lem:spurDifferenzUnendlich}
	There exists a minimal length 
	$\tilde L_0>1$, which depends only on the model parameters, such that for all $L\ge \tilde L_0$, 
	all $T>0$ and all discrete intervals $A \subset \Gamma_{L}$ we have
	\begin{equation}
		\label{eq:FermizuTemeratur1}
		\mathbb{E}\big[\big\| 1_{A}(X)\big( f_{T}(H) - 1_{<0}(H) \big)1_{A}(X)\big\|_1\big]  
		\le 2|A| \,\big[ C^{3} T + \big(2 + C^{3} L^{-1/2} \big) 
		\e^{-L^{-1/2}/T} \big] .
	\end{equation}
\end{lemma}

\begin{proof}
	We recall that, given a bounded measurable function $\zeta: \R \rightarrow\R$ with decomposition 
	$\zeta= \zeta_{+} - \zeta_{-}$ in its positive and negative part, the estimate
	\begin{align}
		\label{Betrag-rein}
 		\| 1_{A}(X) \zeta(H) 1_{A}(X)\|_{1} & \le \| 1_{A}(X) \zeta_{+}(H) 1_{A}(X)\|_{1}
			+ \| 1_{A}(X) \zeta_{-}(H) 1_{A}(X)\|_{1} \notag\\
		& = \Tr\big\{ 1_{A}(X) |\zeta|(H) 1_{A}(X)\big\}
	\end{align}
 	holds. This, ergodicity with respect to $2\Z$-translations and the Pastur--Shubin formula 
	for the integrated density of states $\mathcal{N}(E) = \big(\mathbb{E} [ \langle\delta_{0}, 1_{< E}(H)
	\delta_{0}\rangle] + \mathbb{E} [ \langle\delta_{1}, 1_{< E}(H)\delta_{1}\rangle]\big)/2$	imply
	\begin{align}
		\label{pre-ergod}
		\mathbb{E}&\big[\big\| 1_{A}(X)\big( f_{T}(H) - 1_{<0}(H) \big)1_{A}(X)\big\|_1\big]
		\notag \\
			& \le \mathbb{E}\big[ \Tr\big\{1_{A}(X) |f_{T}(H) - 1_{<0}(H)| 1_{A}(X)\big\}\big] 
			\le 2|A|\int_\R\dd\mathcal N(E)\; |f_{T}(E) - 1_{<0}(E)|.
	\end{align}
	We split the integral over $\R$ into two contributions from $\R_{>0}$, respectively $\R_{<0}$, and 
	only discuss the one from $\R_{>0}$. The other one from $\R_{<0}$ will have the same upper bound.
	Thus, for every $L\in\N$, we infer from partial integration
	\begin{align}
		\label{splitE}
		\int_{0}^{\infty}\dd\mathcal N(E)\;f_{T}(E) & =  \int_{0}^{L^{-1/2}} \dd E\, 
			\big( \mathcal{N}(E) - \mathcal{N}(0)\big) \, (-f_{T})'(E)  \notag \\
			& \quad + \big(\mathcal N(L^{-1/2})-\mathcal N(0)\big) f_T(L^{-1/2}) \notag \\
			& \quad + \int_{L^{-1/2}}^\infty\dd\mathcal{N}(E)\; f_T(E).
	\end{align}
	The integral in the last line of \eqref{splitE} is bounded from above by
	$\e^{-L^{-1/2}/T}$.
	According to Theorem~\ref{thm:Jitomirskaya}\,(ii), there exists $\varepsilon_0>0$, which depends only on $v$ and 
	on the probabilities $p_{\pm}$, 
	such that $\left|\mathcal N(E)-\mathcal N(0)\right|< 2 \mathcal N^\prime(0)|E|$ for all 
	$|E|<\varepsilon_{0}$. From now on we assume that $L > \tilde L_{0} := \varepsilon_{0}^{-2}$.
	Thus, the modulus of the term in the second line of \eqref{splitE} is bounded from above by
	\begin{equation}
 		2 \mathcal{N}'(0) L^{-1/2} \,\e^{-L^{-1/2}/T}  \le 2^{-1} C^{3}  L^{-1/2} \,\e^{-L^{-1/2}/T},
	\end{equation}
	where we used $f_{T} \le \e^{- (\,\pmb\cdot\,)/T}$ and Theorem~\ref{cor:dosbounds}.
	Since $(-f_{T})' \ge 0$, we bound the modulus of the first integral on the right-hand side 
	of \eqref{splitE} from above by
	\begin{align}
		\label{diff-N-use}
 		2 \mathcal{N}^\prime(0) \int_0^{L^{-1/2}}\!\!\!\dd E\, E\, (-f_{T})' (E)
		& = 2 \mathcal{N}^\prime(0) \bigg\{ -L^{-1/2} f_{T}(L^{-1/2})
			+ \int_{0}^{L^{-1/2}} \!\!\!\d E\, f_{T}(E)\bigg\} \notag\\
		&	\le 2 \mathcal{N}^\prime(0) \, T \le 2^{-1} C^{3}T.
	\end{align}
	Collecting the three upper bounds for the contributions to \eqref{splitE}, and adding the identical 
	upper bound for the contribution from the integral over $\R_{<0}$ to \eqref{pre-ergod}, we obtain the claim.
\end{proof}

\begin{lemma}
	\label{lem:Lendlich}
	Let $\alpha >0$. For all $L>L_{\min}$, all $T>0$ 
	and all discrete intervals $A \subset \Gamma_{L}$ we have
	\begin{multline}
		\label{eq:FermizuTermperatur2} 
		\E{\left\|1_{A}(X) \big(1_{<0}(H_L)-f_{T}(H_{L})\big)	1_{A}(X)\right\|_1} \\
		\le |A| \Big[C^{4}T + \frac{2C}{L} + \e^{-L^{-1/2-\alpha}/T} + \e^{-cL^{\alpha/2}}\Big],
	\end{multline}
	where $L_{\min}$, $C$ and $c$ originate from Theorem~\ref{thm:Jitomirskaya}.
\end{lemma}

\begin{proof}
	The principal strategy here is the same as in the proof of Lemma \ref{lem:spurDifferenzUnendlich},
	but instead of ergodicity and regularity of the integrated density of states, we rely on the delocalisation results of 
	Theorem~\ref{thm:Jitomirskaya}. 	Thus, let $L \ge L_{\min}$ and $\omega\in(\Omega_L(\alpha))^c$. 
	We drop $\omega$ from the notation of all quantities in this proof and infer from \eqref{Betrag-rein} that
	\begin{equation}
		\big\| 1_{A}(X)\big(1_{<0}(H_L)-f_{T}(H_{L})\big)	1_{A}(X)\big\|_1 
		\le\sum_{x\in A} \big\langle\delta_{x},  \big|1_{<0}(H_L)-f_{T}(H_{L})\big| \delta_{x}\big\rangle.
	\end{equation}
	Since $|1_{<0}-f_{T}| \le \e^{-|\,\pmb\cdot\,|/T}$, 
	we obtain for all $x\in A$ 
	\begin{align}
		\label{fin-vol-Riemann}
		\big\langle\delta_{x},  \big|1_{<0}(H_L) - f_{T}(H_{L})\big| \delta_{x}\big\rangle 
		&	\le \big\langle\delta_{x}, 1_{\mathcal{W}_L}(H_L) \e^{-|H_L|/T} \delta_{x}\big\rangle
			+ \e^{-L^{-1/2-\alpha}/T} \notag\\
		&	= \sum_{j= J_{<}}^{J_{\ge}} |\langle\delta_{x}, \psi_{E_j}\rangle|^2 \e^{-|E_j|/T}
			+ \e^{-L^{-1/2-\alpha}/T},
	\end{align}
	where ${J}_{<}$ and ${J}_{\ge}$ were defined in \eqref{J-grenze}. 
	Theorem \ref{thm:Jitomirskaya} implies $|\psi_{E_j}(x)|^2\le C/L$ for all $j\in \{J_{<}, \ldots, J_{\ge}\}$
	and $C/L \le (C^4/\pi)|E_j-E_{j\pm 1}|$ for all $j\in \{ J_{<}, \ldots, -2\}$, resp.\ $j\in \{ 1,\ldots, J_{\ge} \}$. 
	This yields the following upper bound for the sum in \eqref{fin-vol-Riemann} 
	\begin{align}
		\frac{C}{L} \sum_{j= J_{<}}^{J_{\ge}}  \e^{-|E_j|/T}
		& \le \frac{C^{4}}{\pi} \sum_{j=J_{<}}^{-2} |E_j-E_{j+1}| \e^{-|E_j|/T}
			+ \frac{C^{4}}{\pi} \sum_{j=1}^{J_{\ge}} |E_j-E_{j-1}| \e^{-|E_j|/T} \notag\\
		& \quad+\frac{C}{L} \, \big( \e^{-|E_{-1}|/T} + \e^{-|E_0|/T}\big) \notag\\
		& \le \frac{C^{4}}{\pi} \int_{-L^{-1/2-\alpha}}^{L^{-1/2-\alpha}}\dd E\, \e^{-|E|/T} 
			+ \frac{2C}{L}
			\le \frac{2C^{4}T}{\pi} + \frac{2C}{L}.
	\end{align}
	Therefore, we conclude 
	\begin{multline}
		\E{\left\|1_{A}(X) \big(1_{<0}(H_L)-f_{T}(H_{L})\big)	1_{A}(X)\right\|_1} \\
		\le |A| \Big[C^{4}T + \frac{2C}{L} + \e^{-L^{-1/2-\alpha}/T} + \Pp\big(\Omega_L(\alpha)\big)\Big]
	\end{multline}
	and deduce the claim with \eqref{prob-bad-event}.
\end{proof}

\begin{proof}[Proof of Lemma~\ref{shift-function}]
 	We combine \eqref{eq:Kreinsformula}, \eqref{eq:Krein}, the triangle inequality and Lemmata \ref{lem:CombesThomas} --
	\ref{lem:Lendlich}. Furthermore, $|A|\e^{-cL^{\alpha/2}} \le 2L \e^{-cL^{\alpha/2}} \le 1 \le C$ 
	for all $L \ge L_{1}$, where 
	the minimal length $L_{1}$ depends only on $\alpha$ and on the model parameters (but not on $A$). We set $L_{0} := \max\{\tilde L_{0},
	L_{\min}, L_{1}\}$.
\end{proof}

\appendix

\section{On the proof of Theorem~\ref{thm:Jitomirskaya}}\label{ch:Appendix}

Theorem~\ref{thm:Jitomirskaya} contains slight improvements of results from \cite{JitomirskayaSchSt}, 
which are necessary to deduce our main result. For us it is vital to control the quantity $C$ in 
Theorem~\ref{thm:Jitomirskaya} in the limit $v\downarrow 0$ which is not done in \cite{JitomirskayaSchSt}. 
Therefore we repeat some arguments of \cite{JitomirskayaSchSt} in this appendix and keep track of the constants. 
Again, we assume $v \in\;]0,2[$ and we restrict ourselves to the case $E_{F}=0$, 
the case of the other critical energy $E_{F}=v$ being analogous.

Given $V \in\{0,1\}$ and $E\in\R$, we define the single-step transfer matrix by
\begin{equation}
	W_V(E):=\left(\begin{array}{@{}cc@{}} vV-E & -1 \\ 1&0 \end{array}\right)\in\R^{2\times 2}.
\end{equation}
The (multi-step) transfer matrix 
\begin{equation}
	\label{multi-W}
	W^\omega(E;y,x) := \left\{\begin{array}{@{}cl@{}} W_{V^\omega(y-1)}(E)\cdots W_{V^\omega(x)}(E) &\text{if }x<y, \\
		1_{2\times 2} &\text{if }x=y, \end{array}\right.
\end{equation}
relates the solution of the discrete Schr\"odinger equation \eqref{diff-eq} at different sites 
\begin{equation}
	W_{V^{\omega}(x)}(E;y,x)\left(\begin{array}{@{}c@{}} \phi^\omega_E(x) \\ \phi^\omega_E(x-1)	\end{array}\right)
	=\left(\begin{array}{@{}c@{}} \phi^\omega_E(y)\\ \phi^\omega_E(y-1) \end{array}\right),
\end{equation}
where $x \le y$. In our model, the single-dimer transfer matrix 
\begin{equation}
	D_{V}(E):=\big(W_{V}(E)\big)^2
\end{equation}
and its similarity transform 
\begin{align}
	\label{single-T}
	T_{V}(E):= M^{-1}_vD_{V}(E)M_v =: 
	\left(\begin{array}{@{}cc@{}} \overline {a_V(E)} & b_V(E) \\ \overline {b_V(E)} & a_V(E) \end{array}\right)
\end{align}
with entries $a_{V}(E), b_{V}(E) \in\C$ are of great relevance. Here, the change of basis in $\C^{2}$ induced by 
\begin{equation}
	M_v:=m_{v}\left(\begin{array}{@{}cc@{}} \frac{1}{2}\left(v - \i\sqrt{4-v^2}\right)
		&\frac{1}{2}\left(v + \i \sqrt{4-v^2}\right) \\ 1&1 \end{array}\right)
\end{equation}
simultaneously diagonalises $D_{0}(0)$ and $D_{v}(0)$, i.e.\ $T_{0}(0)=-1_{2\times 2}$ and $T_{v}(0)$ are both diagonal.
The real parameter $m_{v}>0$ is chosen such that $|\det M_{v}| =1$.
We remark that for every $w\in\R^{2}$ there exists $z\in\C$ such that 
\begin{equation}
 	M_{v}^{-1} w = \left(\begin{array}{@{}c@{}} z \\ \overline{z} \end{array}\right).
\end{equation}
In analogy to \eqref{multi-W}, we define the modified (multi-step) dimer transfer matrix as
\begin{equation}
	\label{multi-T}
	T^\omega(E;y,x) := \left\{\begin{array}{cl} T_{V^\omega(y-2)}(E)\cdots T_{V^\omega(x)}(E) &\text{if }x<y, \\
		1_{2\times 2} & \text{if }x=y, \end{array}\right.
\end{equation}
where $x,y\in 2\Z$.

For later usage we state the Taylor expansions of the entries of $T_{V}(E)$ as $E\downarrow 0$
\begin{equation}
	\label{eq:ab}
	\begin{split}
		a_0(E) & =-1 - E \;\frac{2\i}{\sqrt{4-v^2}}+\mathcal O(E^2),\\
		a_1(E) & =-1+\frac{v^2}{2} + \i \;\frac{v}{2} \;\sqrt{4-v^2} - E \;\Big(v+\i\;\frac{2-v^2}{\sqrt{4-v^2}}\Big)
			+\mathcal O(E^2),\\
		b_0(E) & =\frac{Ev}{2}\Big(-1 +\i\;\frac{v}{\sqrt{4-v^2}}\Big)+\mathcal O(E^2),\\
		b_1(E) & =-b_0(E)+\mathcal O(E^2).\end{split}
\end{equation}

\begin{lemma}[Cf.\ (42) in \cite{JitomirskayaSchSt}]
	\label{lem:ap1}
	Given $\theta\in[0,2\pi[$\,, let $e_\theta:=\frac{1}{\sqrt 2}(\e^{-\i\theta},\e^{\i\theta})^T$.
	For all $v\in\;]0,2[$\,, $V\in\{0,1\}$ and all $E\in\R$ there exists maps $\Theta_V:\;[0,2\pi[\; 
	\rightarrow\R$ and 
	$\rho_V:\;[0,2\pi[\;\rightarrow\;]0,\infty[\,$ such that
	\begin{equation}
		\label{WinkelEntwicklung}
		T_{V}(E) e_\theta = \rho_V(\theta) \, e_{\Theta_V(\theta)}
	\end{equation}
	for all $\theta\in[0,2\pi[$\,. Furthermore, we have
	\begin{equation}
		\label{eq:rhoEntw}
		\rho^2_V(\theta)=1+2|b_V(E)|^2+2\Re\big(a_V(E)b_V(E)\e^{2\i\theta}\big).
	\end{equation}
\end{lemma}

\begin{proof}
	The form of $T_{V}(E)$ in \eqref{single-T} implies that for every non-zero $w_z:=(z,\overline z)^T$, 
	$z\in\C\backslash\{0\}$  there exists $\zeta\in\C\backslash\{0\}$ such that $T_{V}(E) w_{z} = w_{\zeta}$.
	Since $w_{\zeta}= \rho e_\Theta$ for a unique $\rho>0$ and $\Theta\in[0,2\pi[\,$, the first part of the lemma follows.
	The equality \eqref{eq:rhoEntw} is verified by a direct computation. 	
\end{proof}

In the following lemma, which is a modification of (49) in \cite{JitomirskayaSchSt}, 
we use the notation $|\pmb\cdot|$ for the Euclidean norm on $\C^{2}$.

\begin{lemma}
	Let $v\in\;]0,2[\,$, $L\in\N$, $E\in [-v,v]$, $\omega\in\Omega$ and $x,y\in\Gamma_{L}$ with $x\le y$.
	\begin{nummer}
	\item  
		\label{Zwischenresultat}
		Then there exists a constant $c_v \in \;]0,\infty[$\,, which depends only on $v$ and obeys
		\begin{equation}
			\label{eq:cv-lim}
			\lim_{v\downarrow 0}c_v=0,
		\end{equation} 
		such that for all unit vectors
		$w\in\R^2$, $|w|=1$, there is an angle $\xi_w\in[0,2\pi[$ such that 
	 	\begin{equation}\label{eq:JSBSwithC}
			\ln \big(|W^\omega(E;x,-L)w|^{2}\big) \in 2E \sum_{k=k_0}^{k_1-1} \Re\big( d_{V^\omega(2k)}\e^{2\i\vartheta_k}\big)
			+\mathcal O(E^2L) + c_{v} [-1,1]
		\end{equation}
		with $d_V:= {a_{V}}(0) b_{V}'(0)$ for $V\in\{0,1\}$ and where  
		\begin{equation}
			\label{eq:Zwischenresultat}
 			k_0:=\min\{k\in\Z:-L\le 2k\}, \qquad k_1:=\max\{k:2k\le x\}, \\
		\end{equation}
		$\vartheta_{k_0}:=\xi_w$ and $\vartheta_{k+1}=\Theta_{V^\omega(2k)}(\vartheta_k)$ for all $k\in\{k_0,\hdots,k_1-1\}$.
		The $\mathcal O(E^2L)$-term in \eqref{eq:JSBSwithC} has no further dependencies except on the model parameters.
	\item
		\label{lem:Zwischenresultat2} 
		Let $\{w_1,w_2\}$ be an orthonormal basis of $\R^2$. Then
		\begin{equation}
			\|W^\omega(E;y,x)\|\le 2\max_{w\in\{w_1,w_2\}}\max_{z\in\Gamma_L}|W^\omega(E;z,-L)w|^2.
		\end{equation}
	\end{nummer}
\end{lemma}

\begin{proof}
	\begin{nummer}
	\item
	For all $x\in\Gamma_L$  we have
	\begin{equation}
		W^\omega(E;x,-L)=W^\omega(E;x,2k_1)M_v  T^\omega(E;2k_1,2k_0)M_v^{-1} W^\omega(E;2k_0,-L).
	\end{equation}
	For $w\in\R^2$, $|w|=1$, let the angle $\xi_w\in[0,2\pi[$ be given as the unique solution of
	\begin{equation}
		e_{\xi_w}=M_v^{-1} W^\omega(E;2k_0,-L)w/|M_v^{-1} W^\omega(E;2k_0,-L)w|.\label{eq:Winkedef}
	\end{equation}
	We claim that 
	\begin{equation}
		\label{transferrestimate}
		\ln|W^\omega(E;x,-L)w|^2 \in  \sum_{k=k_0}^{k_1-1}\ln\big(\rho_{V^\omega(2k)}(\vartheta_k)^{2}\big) + c_v [-1,1]
	\end{equation}
	with
	\begin{equation}
		\label{def:cv}
		c_v:= 4\ln\big(\|M_v\|\big) + 4 \max_{E\in[-v,v]}\max_{V\in\{0,1\}}\ln \|W_{V}(E)\|  >0.
	\end{equation}
	To see the validity of \eqref{transferrestimate}, we iterate Lemma~\ref{lem:ap1} and conclude
	\begin{align}
		\label{eq:iterate}
		|W^\omega(E;x,-L)w| & = |W^\omega(E;x,2k_{1}) M_{v}e_{\vartheta_{k_{1}}}|
	 		\prod_{k=k_0}^{k_1-1} \rho_{V^\omega(2k)}(\vartheta_k) \notag\\
		& \qquad \times |M_v^{-1} W^\omega(E;2k_0,-L)w|.
	\end{align}
	Furthermore, we note that 
	\begin{equation}
		\label{2x2-mat}
 		\|A^{-1}\| = \|A\| \qquad \text{and} \qquad \frac{1}{\|A\|} \le |Aw| \le \|A\|
	\end{equation}
	hold for any complex $2\times2$-matrix $A$ with $|\det A| =1$ and any $w\in\C^{2}$ with $|w|=1$. Applying this to the 
	first and last factor on the right-hand side of \eqref{eq:iterate}, yields \eqref{transferrestimate}.

	The estimate \eqref{eq:JSBSwithC} now follows from \eqref{eq:rhoEntw} and a Taylor expansion in the energy $E$, using 
	\eqref{eq:ab}.
	Since $\max_{E\in [-v,v]}\|W_V(E)\|,\, \|M_v\| \rightarrow 1$ as $v\rightarrow 0$ for every $V\in\{0,1\}$, 
	we conclude \eqref{eq:cv-lim} from \eqref{def:cv}.
	\item 
		For all $x,\;y\in\Gamma_L$ we have
		\begin{equation}
			\|W^\omega(E;y,x)\| \le \|W^\omega(E;y,-L)\| \|W^\omega(E;x,-L)^{-1}\|
			\le \max_{z\in\Gamma_L}\|W^\omega(E;z,-L)\|^2,
		\end{equation}
		where we used the equality of norms in \eqref{2x2-mat}.
		The claim follows from the observation that for any $2\times 2$ matrix 
		\begin{equation}
			\|A\|^2\le2\max_{w\in\{w_1,w_2\}}\|Aw\|^2.
		\end{equation}
	\end{nummer}
\end{proof}

The next lemma accounts for a perturbation in energy and is a variation of \cite[Lemma 2.1]{DamanikTch03} or 
\cite[Thm.~2J]{Simon1996}. 

\begin{lemma}
	\label{lem:exp-Reihe}
	Let $E,\varepsilon\in\R$, $\omega\in\Omega$, $L\in\N$ and $G_{E}^{\omega}:=\max_{x,y\in\Gamma_L, x<y}\|W^\omega(E;y,x)\|$. 
	Then, we have for all $x\in\Gamma_L$ and all $w\in\R^2$ with $|w|=1$ the estimate
	\begin{equation}
		|W^\omega(E+\epsilon;x,-L)w|^{2} \in |W^\omega(E;x,-L)w|^{2} +  (G_{E}^{\omega})^{2} \big(\e^{4L|\epsilon|G_{E}^{\omega}} - 1\big) \;[-1,1] .
	\end{equation}
\end{lemma}

\begin{proof}
	For $V\in\{0,1\}$ and $E,\varepsilon \in\R$ we observe  
	\begin{equation}
		W_{V}(E+\epsilon) = W_{V}(E) - \epsilon\left(\begin{array}{@{}c@{\;\,}c@{}} 1 & 0 \\	0 & 0 \end{array}\right)
	\end{equation}
	and expand $W^\omega(E+\epsilon;x,-L)$ in powers of $\epsilon$. For the upper bound, this leads to the estimate
	\begin{align}
		|W^\omega(E+\epsilon;x,-L)w|
		& \le |W^\omega(E;x,-L)w| + G_{E}^{\omega} \max_{x\in\Gamma_L}\sum_{j=1}^{x+L} \tbinom{x+L}{j} \; (|\epsilon| G_{E}^{\omega})^j \notag\\
		& \le |W^\omega(E;x,-L)w| + G_{E}^{\omega} \sum_{j=1}^{|\Gamma_L|} \frac{\big(|\Gamma_L||\epsilon| G_{E}^{\omega}\big)^j}{j!}\notag\\
		& \le |W^\omega(E;x,-L)w| + G_{E}^{\omega} \big(\e^{2L|\epsilon|G_{E}^{\omega}} - 1\big)
	\end{align}
	for all $x\in\Gamma_L$ and all unit vectors $w\in\R^{2}$. For the lower bound, we use the inverse triangle inequality to estimate the expansion in $\varepsilon$ 
	according to
	\begin{align}
		|W^\omega(E+\epsilon;x,-L)w|
		& \ge |W^\omega(E;x,-L)w| - G_{E}^{\omega} \max_{x\in\Gamma_L}\sum_{j=1}^{x+L} \tbinom{x+L}{j}\;(|\epsilon| G_{E}^{\omega})^j \notag\\
		& \ge |W^\omega(E;x,-L)w| - G_{E}^{\omega} \big(\e^{2L|\epsilon|G_{E}^{\omega}} - 1\big)		
	\end{align}
	for all $x\in\Gamma_L$ and all unit vectors $w\in\R^{2}$. We note that for any $a,b,c \ge 0$, the two-sided estimate
	$a \in b+ c \, [-1,1]$ implies $a^{2} \in b^{2}+ c(2b+c) \, [-1,1]$. In our case, we have $b:=|W^\omega(E;x,-L)w| \le G_{E}^{\omega}$,
	which implies the claim.
\end{proof}

For convenience we quote \cite[Thm.~6]{JitomirskayaSchSt} in our notation and note that the assumption 
$|\langle\e^{2\i \eta_{\pm}}\rangle| < 1$ there is always fulfilled in the dimer model considered in this paper. 

\begin{theorem}[\protect{\cite[Thm.~6]{JitomirskayaSchSt}}]
	\label{thm6-Jito}
	Let $v\in\, ]0,2[\,$. For $L\in\N$, $\alpha>0$, $\theta\in[0,2\pi[$ and $E\in\mathcal{W}_{L}$ let 
	\begin{equation}
		\Omega_L(\alpha,E,\theta):=\Big\{\omega\in\Omega:\;\exists\, k_1\in\tfrac{1}{2}\Gamma_L \cap \Z\textrm{ such that }
		\Big|\sum_{k=k_0}^{k_1}d_{V^\omega(2k)}\e^{2\i\vartheta_k}\Big|\ge L^{\alpha+\frac12}\Big\},
	\end{equation}
	with $d_{V}$, $k_0$ and $\vartheta_k$ defined as in Lemma~\ref{Zwischenresultat} with $\vartheta_{k_0}=\theta$.
	Then there exists quantities $C_1,\;C_2>0$, depending only on $\alpha$ and the model parameters, such that
	\begin{align}
		\mathbb{P}\big(\Omega_L(\alpha,E,\theta)\big)\le C_1 \e^{-C_2L^\alpha}.
	\end{align}
\end{theorem}

\begin{lemma}
	\label{lem:almostthere}
	Let $v\in\;]0,2[\,$. For all $\alpha>0$ there exists $L_0\in\N$ such that for all $L\ge L_0$ there exists a measurable subset
	$\Omega_L(\alpha)\subseteq\Omega$ and a constant $c>0$ such that
	\begin{equation}
		\mathbb P\big(\Omega_L(\alpha)\big)\le \e^{-cL^{\alpha/2}}
	\end{equation}
	and such that for all $\omega\in\big(\Omega_L(\alpha)\big)^c$, ${E \in\mathcal{W}_{L}}$ and ${x\in\Gamma_L}$
	\begin{equation}
		\Big|W^\omega(E;x,-L)  \Big(\begin{array}{@{}c@{}} \hbox{\footnotesize$1$} \\[-.5ex] 
		\hbox{\footnotesize$0$} \end{array} \Big) \Big|^2
		\in [ \e^{-3c_v}, \e^{3c_v}],
	\end{equation}
	where the constant $c_{v}$ is given by Lemma~\ref{Zwischenresultat}, see \eqref{def:cv}.
\end{lemma}

\begin{proof}
	Let $w_1:=(1,0)^{T}$ and $w_2:=(0,1)^{T}$. In view of \eqref{eq:Winkedef}, we define a set of 
	modified Pr\"ufer angles
	\begin{align}
		\Xi:=\Big\{\xi\in[0,2\pi[\,:\;\exists W\in \{1_{2\times 2}, W_{0}(E), W_{1}(E)\}, 
			& \,w\in\{w_1,w_2\}	\notag\\
			& \textrm{ with } e_{\xi} =\frac{M_v^{-1} Ww}{|M_v^{-1} Ww|}\Big\}
	\end{align}
	with cardinality $|\Xi| \le 6$. Let 
	\begin{equation}
		\Omega_L(\alpha,E):=\bigcup_{\theta\in \Xi}\Omega_L(\alpha/2,E,\theta).
	\end{equation} 
	Hence, $\mathbb P(\Omega_L(\alpha,E))\le 6C_1\e^{-C_2L^{\alpha/2}}$ by Theorem~\ref{thm6-Jito}.
	We assume $L \ge v^{-2}$ so that $\mathcal{W}_{L} \subset [-v,v]$. Thus, for all $E\in \mathcal{W}_{L}$ and 
	$\omega\in(\Omega_L(\alpha,E))^c$ the estimate \eqref{eq:JSBSwithC} yields
	\begin{equation}
		\ln \big(|W^\omega(E;x,-L)w|^{2}\big) \in \mathcal O(E^2L) + \big(2E L^{1/2+\alpha/2} + c_{v}\big) \; [-1,1]
	\end{equation}
	for all $x\in\Gamma_{L}$ and $w\in\{ w_{1}, w_{2}\}$.
Hence there exists $L_0^\prime \ge v^{-2}$ such that for all $L\ge L_0^\prime$, all $E\in\mathcal{W}_{L}$,
all $\omega\in(\Omega_L(\alpha,E))^c$, all $x\in\Gamma_{L}$ and $w\in\{w_{1},w_{2}\}$, we have
\begin{equation}
	\label{eq:lnWC}
	\ln\big(|W^\omega(E;x,-L)w|^2\big) \in 2c_v \, [-1,1].
\end{equation}
The upper bound in \eqref{eq:lnWC} and Lemma \ref{lem:Zwischenresultat2} imply for the quantity $G_{E}^{\omega}$ in 
Lemma~\ref{lem:exp-Reihe}
\begin{equation}
	\label{eq:normWC}
	G_{E}^{\omega}=\max_{x,y\in\Gamma_L, x<y} \|W^\omega(E;y,x)\|\le 2 \e^{2c_v}
\end{equation}
for all $\omega\in(\Omega_L(\alpha,E))^c$.
We define
\begin{equation}
	\Omega_L(\alpha):=\bigcup_{\substack{n \in\Z: \\[.5ex] n/L^{2} \in\mathcal{W}_{L}}}\Omega_L(\alpha,n/L^2).
\end{equation}
Hence there exists $L_0^{\prime\prime} \ge L_{0}'$ and $c>0$ such that for all $L\ge L_0^{\prime\prime}$ we have 
\begin{equation}
	\mathbb P\big[\Omega_L(\alpha) \big] \leq 18L^{3/2}C_1\e^{-C_2L^{\alpha/2}}\le \e^{-cL^{\alpha/2}}.
\end{equation}
Now, we consider $\omega\in(\Omega_L(\alpha))^c$ arbitrary and $n\in\Z$ arbitrary such that $E_{n} :=n/L^2 \in\mathcal{W}_{L}$ 
and apply Lemma~\ref{lem:exp-Reihe}, \eqref{eq:lnWC} 
and \eqref{eq:normWC} with $E_{n}$, $w=w_{1}$ and any $|\varepsilon| \le L^{-2}$. This gives 
\begin{align}
	|W^\omega(E;x,-L)w_{1}|^{2} & \in |W^\omega(E_{n};x,-L)w_{1}|^{2} + 4 \e^{4c_{v}}\big(\exp\{4\e^{2c_{v}}/L\} -1 \big) \; [-1,1] \notag\\
	& \subseteq [\e^{-2c_{v}},e^{2c_{v}}] + 4 \e^{4c_{v}}\big(\exp\{4\e^{2c_{v}}/L\} -1 \big) \; [-1,1]
\end{align}
for all $x\in\Gamma_{L}$ and all $E\in D_{n}:= E_{n} + L^{-2} [-1,1]$.
Since
\begin{equation}
	\mathcal{W}_{L} \subseteq \bigcup_{n \in\Z:\;  E_{n}\in\mathcal{W}_{L}} D_n
\end{equation}
there exists $L_0 \ge L_0^{\prime\prime}$ such that for all $L \ge L_0$, all $\omega\in(\Omega_L(\alpha))^c$, all $E\in\mathcal{W}_{L}$
and all $x\in\Gamma_{L}$ we have
\begin{equation}
	|W^\omega(E;x,-L)w_1|^2 \in  \, [\e^{-3c_v},\e^{3c_v}].
\end{equation}
\end{proof}

\begin{proof}[Proof of Theorem \ref{thm:Jitomirskaya}]
	\begin{nummer}
	\item
		Let us first proof \eqref{eq:GleichVerteilung}. 
		For every $L\in\N$, $x\in\Gamma_L$, $E\in\R$ and $\omega\in\Omega$, we infer from \eqref{diff-eq} that
		\begin{equation}
			\label{eq:RadienAbschaetzung}
			r_{x}^\omega(E)^2 = \phi^\omega_E(x)^2+\phi^\omega_E(x-1)^2
			= \Big|W^\omega(E;x,-L)
				\Big(\begin{array}{@{}c@{}} \hbox{\footnotesize$1$} \\[-.5ex] \hbox{\footnotesize$0$} \end{array} \Big)\Big|^2 
				/(R^{\omega}_{E})^2,	
		\end{equation}
		with the normalisation
		\begin{equation}
			(R^{\omega}_{E})^2 :=\sum_{k=0}^{L-1}\Big|W^\omega(E;-L+1+2k,-L) \Big(\begin{array}{@{}c@{}} \hbox{\footnotesize$1$} \\[-.5ex] 
			\hbox{\footnotesize$0$} \end{array} \Big)\Big|^2.
		\end{equation}
		Given $\alpha >0$, Lemma~\ref{lem:almostthere} provides the existence of a minimal length $L_{0}\in\N$ such that for all
		$L \ge L_0$, $\omega\in(\Omega_L(\alpha))^c$, $x\in\Gamma_{L}$ and $E\in \mathcal{W}_{L}$, the two-sided estimate 
		\begin{equation}
			\label{eq:Wbounds}
			(R^{\omega}_{E})^2 \in [L\e^{-3c_v},L\e^{3c_v}]
		\end{equation}
		holds.  
		Thus, \eqref{eq:Wbounds}, another application of Lemma~\ref{lem:almostthere} and \eqref{eq:RadienAbschaetzung} yield
		\eqref{eq:GleichVerteilung} with the constant 
		\begin{equation}
			\label{Constant!!}
			C = \e^{6c_v},
		\end{equation}
		and \eqref{eq:cv-lim} implies \eqref{eq:Cnach1}.

		To prove the level-spacing estimate \eqref{eq:LevelSpacing}, let $L_{0}$ be as above, $L \ge L_0$, $\omega\in(\Omega_L(\alpha))^c$
		and let $E,E'\in \mathcal{W}_{L}$ be two adjacent eigenvalues of $H_L^\omega$ with $E<E'$.
		For $E^{(\prime)}$ to be an eigenvalue, Dirichlet boundary conditions $\phi^{\omega}_{E^{(\prime)}}(L)=0$ have to be met on the 
		right-hand side of $\Gamma_{L}$, that is,
		$\theta^\omega_{L}(E^{(\prime)})\in\pi/2+\pi\Z$. Since $\theta^\omega_{L}$ is a continuous, increasing function with 
		respect to the energy, $E$ and $E'$ are adjacent eigenvalues if and only if the Pr\"ufer-angle difference satisfies 
		$\theta_{L}^\omega(E^\prime)-\theta_{L}^\omega(E) =\pi$. Using \eqref{eq:WinkelgeschwindigkeitGamma}, this can be rewritten as
		\begin{equation}
			\label{eq:pi-int}
			\pi = \int_E^{E^\prime}\dd\epsilon\, \frac{\d}{\d\varepsilon} \theta^\omega_{L}(\epsilon)
			= \int_E^{E^\prime}\dd \epsilon \sum_{x=-L}^{L-1} \left(\frac{\phi_\epsilon^\omega(x)}{r_L^\omega(\epsilon)} \right)^{2}
			= \int_E^{E^\prime}\dd \epsilon \; \frac{1}{\big(r_L^\omega(\epsilon)\big)^{2}}.
		\end{equation}
		The eigenfunction estimate \eqref{eq:GleichVerteilung} does not apply directly to $r_L^\omega(\epsilon)$ for 
		$\varepsilon\in\mathcal{W}_{L}$, but only after an 
		additional iteration with the transfer matrix
		\begin{equation}
 			\big(r_L^\omega(\epsilon)\big)^{2} = \Big| W_{V^{\omega}(L-1)}(\varepsilon) 
			\Big(\begin{array}{@{}c@{}} \hbox{\footnotesize$\cos\theta^{\omega}_{L-1}(\varepsilon)$} \\ 
			\hbox{\footnotesize$\sin\theta^{\omega}_{L-1}(\varepsilon)$} \end{array} \Big)\Big|^{2} \big(r_{L-1}^\omega(\epsilon)\big)^{2}.
		\end{equation}
		We already have $\big(r_{L-1}^\omega(\epsilon)\big)^{2} \in L^{-1} [C^{-1},C]$ for every $\omega\in(\Omega_L(\alpha))^c$ 
		by \eqref{eq:GleichVerteilung}. Since 
		$ \max_{V\in\{0,1\}} \|W_{V}(\varepsilon)\| \le \e^{c_{v}/4} \le C$ uniformly in $\varepsilon\in\mathcal{W}_{L}$ 
		by \eqref{def:cv}, we deduce from \eqref{2x2-mat} that $\big(r_L^\omega(\epsilon)\big)^{2} \in L^{-1} [C^{-3},C^{3}]$.
		Inserting this into \eqref{eq:pi-int}, yields 
		\begin{equation}
			E^\prime-E\in\frac{\pi}{L} \; [C^{-3},C^{3}].
		\end{equation}
	\item
 		The existence of the density of states $\mathcal{N}'(E_{c})$ follows from \cite[Thm.~3]{JitomirskayaSchSt}, the upper and lower bounds 
		from Dirichlet--Neumann bracketing and the eigenvalue spacing in the critical energy window, as we show now. 
		
		For $L\in\N$ we introduce the restricted Schrödinger operators $H^{\omega,\,D/N}_L$ with 
		Dirichlet, respectively Neumann, boundary conditions
		\begin{equation}
			\begin{split}
				H^{\omega,\,D}_L:=H^{\omega}_L+\proj{\delta_{-L}}{\delta_{-L}}+\proj{\delta_{L-1}}{\delta_{L-1}},\\
				H^{\omega,\,N}_L:=H^{\omega}_L-\proj{\delta_{-L}}{\delta_{-L}}-\proj{\delta_{L-1}}{\delta_{L-1}}.
			\end{split}
		\end{equation}
		Their integrated densities of states at energy $E\in\R$ are given by
		\begin{equation}
			\mathcal N_L^{\omega,\, D/N}(E):=\Tr\Big\{1_{\leq E}\big(H^{\omega,\,D/N}_L\big)\Big\}.
		\end{equation}
		Since $H^{\omega,\,D/N}_L$ are rank-2-perturbations of $H^\omega_L$, the min-max-principle implies
		\begin{equation}
			\label{eq:Finitrank}
			\mathcal N_L^{\omega,\, D/N}(E)\in \Tr\big\{1_{\le E}(H^\omega_L)\big\} + [-2,2].
		\end{equation}
		According to \cite[p.~312]{carlac1990random} Dirichlet--Neumann bracketing yields
		\begin{equation}
			\label{eq:DNbrack}
			\frac{1}{|\Gamma_L|}\mathbb E\big[\mathcal N_L^{D}(E)\big] 
			\le \mathcal N(E)\le\frac{1}{|\Gamma_L|}\mathbb E\big[\mathcal N_L^{N}(E)\big]
		\end{equation}
		for every $E\in\R$ and every $L\in\N$. Thus, we conclude from \eqref{eq:Finitrank} and \eqref{eq:DNbrack} that
		\begin{equation}
			\mathcal N(E+\epsilon)-\mathcal N(E-\epsilon) 
			\in \frac{1}{|\Gamma_L|} \mathbb E\big\{ \Tr 1_{]E-\epsilon,E+\epsilon]}(H_L)\big\} + \frac{4}{|\Gamma_L|} \; [-1,1]			
			\label{eq:densityDiff}
		\end{equation}
		for every $\varepsilon>0$. For fixed $\alpha \in\;]0,1/2[$ let $\epsilon_L:=L^{-1/2-\alpha}$ be half 
		the width of the critical energy window
		$\mathcal{W}_{L}$ around $E_{c} \in\{0,v\}$. 
		Theorem~\ref{thm:Jitomirskaya}(i) provides the existence of a minimal length $L_{0}\in\N$ 
		such that for all $L\ge L_0$ and
		all $\omega\in(\Omega_L(\alpha))^c$ and  we have
		\begin{equation}
			\label{eq:CountingEnergy}
			 \frac{2\epsilon L}{\pi C^3}  - 1 \le \Tr\big\{1_{]E_{c}-\epsilon,E_{c}+\epsilon]}(H^\omega_L)\big\}
			 \le \frac{2\epsilon C^3 L}{\pi} + 1.
		\end{equation}
		The estimates \eqref{eq:CountingEnergy} and \eqref{eq:densityDiff} imply for $L\ge L_{0}$
		\begin{align}
			\frac{\mathcal N(E_{c} + \epsilon_L) - \mathcal N(E_{c} - \epsilon_L)}{2\epsilon_L}
			& \in \frac{1}{ 2\epsilon_L |\Gamma_L|} \; \mathbb{E} \big\{ 1_{\Omega_{L}(\alpha)} 
				\Tr 1_{]E_{c}-\epsilon_{L},E_{c}+\epsilon_{L}]}(H_L)	\big\}	\notag\\
			& \quad + \frac{\mathbb{P}\big\{\big(\Omega_{L}(\alpha)\big)^{c}\big\}}{2\epsilon_L |\Gamma_L|} \; 
					\Big[\frac{2\epsilon_{L} L}{\pi C^{3}} - 1, \frac{2\epsilon_{L} C^3 L}{\pi} + 1\Big] \notag \\
			& \quad + \frac{2}{\epsilon_L |\Gamma_L|} \;[-1,1] \notag\\
			& \subseteq \frac{\mathbb{P}\big\{\big(\Omega_{L}(\alpha)\big)^{c}\big\} L}{\pi|\Gamma_L|} \; [C^{-3}, C^{3}] \notag\\
			&	\quad + \Big( \frac{\mathbb{P}\{\Omega_{L}(\alpha)\}}{2\varepsilon_{L}} + \frac{3}{\varepsilon_{L}|\Gamma_{L}|}\Big) \;
					[-1,1].
		\end{align}
		Now, the claim follows in the limit $L\to\infty$.
	  \qedhere
	\end{nummer}
\end{proof}
%
%
%
%
%
We finish with an elementary auxiliary result needed in the proof of Theorem~\ref{thm:main}.

\begin{lemma}
	\label{lem:lengthcount}
 	Let $\gamma\in\;]0,1/2[$ and $\gamma_L:=\lfloor(\gamma+\gamma^2)L\rfloor-\lfloor\gamma L\rfloor$ for $L\in\N$. 
	Then for all $L^\prime\in\N$ there exists $L\in \N$ such that $L^\prime=\gamma_L$.
\end{lemma}

\begin{proof}
	As $\gamma<1$ and $\gamma+\gamma^2<1$, we infer
	$\lfloor (\gamma+\gamma^2)(L+1)\rfloor-\lfloor (\gamma+\gamma^2)L\rfloor\in\{0,1\}$ and 
	$\lfloor \gamma(L+1)\rfloor-\lfloor \gamma L\rfloor\in\{0,1\}$ for all $L\in\N$.
	Thus, we have $\gamma_{L+1}-\gamma_L\in\{-1,0,1\}$ for all $L\in\N$. Together with 
	$\gamma_1=0$ and $\lim_{L\rightarrow\infty}\gamma_L=\infty$, this yields the claim.
\end{proof}

\newcommand{\noopsort}[1]{} \newcommand{\singleletter}[1]{#1}


\begin{thebibliography}{ARNSS17}
\providecommand{\url}[1]{{\tt #1}}
\providecommand{\urlprefix}{URL }
\providecommand{\eprint}[2][]{e-print {#2}}

\bibitem[AR18]{MR3777285}
H.~Abdul-Rahman, Entanglement of a class of non-{G}aussian states in disordered
  harmonic oscillator systems, {\em J. Math. Phys.\/} {\bf 59}, 031904-1--17
  (2018).

\bibitem[ARNSS17]{MR3671049}
H.~Abdul-Rahman, B.~Nachtergaele, R.~Sims and G.~Stolz, Localization properties
  of the disordered {XY} spin chain: a review of mathematical results with an
  eye toward many-body localization, {\em Ann. Phys.\/} {\bf 529},
  1600280-1--17 (2017).

\bibitem[ARS15]{AbdulRhamanStolz15}
H.~Abdul-Rahman and G.~Stolz, A uniform area law for the entanglement of
  eigenstates in the disordered {XY} chain, {\em J. Math. Phys.\/} {\bf 56},
  121901-1--25 (2015).

\bibitem[AW15]{AizWarBook}
M.~Aizenman and S.~Warzel, {\em Random operators: Disorder effects on quantum
  spectra and dynamics\/}, Graduate Studies in Mathematics, vol. 168, Amer.
  Math. Soc., Providence, RI, 2015.

\bibitem[AFOV08]{Amico:2008en}
L.~Amico, R.~Fazio, A.~Osterloh and V.~Vedral, Entanglement in many-body
  systems, {\em Rev. Mod. Phys.\/} {\bf 80}, 517--576 (2008).

\bibitem[BSW18]{2018arXiv181209144B}
V.~{Beaud}, J.~{Sieber} and S.~{Warzel}, {Bounds on the bipartite entanglement
  entropy for oscillator systems with or without disorder},
  {\em J. Phys. A: Math. Theor.\/} {\bf 52},  235202-1--24 (2019).

\bibitem[BW18]{MR3749589}
V.~Beaud and S.~Warzel, Bounds on the entanglement entropy of droplet states in
  the {XXZ} spin chain, {\em J. Math. Phys.\/} {\bf 59}, 012109-1--11 (2018).

\bibitem[Bek73]{PhysRevD.7.2333}
J.~D. Bekenstein, Black holes and entropy, {\em Phys. Rev. D\/} {\bf 7},
  2333--2346 (1973).

\bibitem[Bek04]{Bekenstein04}
J.~D. Bekenstein, Black holes and information theory, {\em Contemp. Phys.\/}
  {\bf 45}, 31--43 (2004).

\bibitem[BH15]{MR3296162}
F.~G. S.~L. Brand{\~{a}}o and M.~Horodecki, Exponential decay of correlations
  implies area law, {\em Commun. Math. Phys.\/} {\bf 333}, 761--798 (2015).

\bibitem[CC09]{MR2566332}
P.~Calabrese and J.~Cardy, Entanglement entropy and conformal field theory,
  {\em J. Phys. A\/} {\bf 42}, 504005-1--36 (2009).

\bibitem[CKM87]{CarmonaKleinMar87}
R.~Carmona, A.~Klein and F.~Martinelli, Anderson localization for {B}ernoulli
  and other singular potentials, {\em Commun. Math. Phys.\/} {\bf 108}, 41--66
  (1987).

\bibitem[CL90]{carlac1990random}
R.~Carmona and J.~Lacroix, {\em Spectral theory of random {S}chr\"odinger
  operators\/}, Birkh\"auser, Boston, 1990.

\bibitem[DT03]{DamanikTch03}
D.~Damanik and S.~Tcheremchantsev, Power-law bounds on transfer matrices and
  quantum dynamics in one dimension, {\em Commun. Math. Phys.\/} {\bf 236},
  513--534 (2003).

\bibitem[DBG00]{DeBievre:2000tq}
S.~De~Bievre and F.~Germinet, {Dynamical localization for the random dimer
  model}, {\em J. Stat. Phys.\/} {\bf 98}, 1135--1148 (2000).

\bibitem[ECP10]{RevModPhys.82.277}
J.~Eisert, M.~Cramer and M.~B. Plenio, Area laws for the entanglement entropy,
  {\em Rev. Mod. Phys.\/} {\bf 82}, 277--306 (2010).

\bibitem[EPS17]{ElgartPasturShcherbina2016}
A.~Elgart, L.~Pastur and M.~Shcherbina, Large block properties of the
  entanglement entropy of free disordered {F}ermions, {\em J. Stat. Phys.\/}
  {\bf 166}, 1092--1127 (2017).

\bibitem[FS18]{MR3799968}
C.~Fischbacher and G.~Stolz, Droplet states in quantum {XXZ} spin systems on
  general graphs, {\em J. Math. Phys.\/} {\bf 59}, 051901-1--28 (2018).

\bibitem[Has07]{Hastings1D07}
M.~B. Hastings, An area law for one-dimensional quantum systems, {\em J. Stat.
  Mech.\/} {\bf 2007}, P08024-1--14 (2007).

\bibitem[HLS11]{Helling:2011gr}
R.~Helling, H.~Leschke and W.~Spitzer, A special case of a conjecture by
  {W}idom with implications to fermionic entanglement entropy, {\em Int. Mat.
  Res. Not.\/} {\bf 2011}, 1451--1482 (2011).

\bibitem[HHHH09]{Horodecki:2009gb}
R.~Horodecki, P.~Horodecki, M.~Horodecki and K.~Horodecki, Quantum
  entanglement, {\em Rev. Mod. Phys.\/} {\bf 81}, 865--942 (2009).

\bibitem[ISL12]{PhysRevB.85.094417}
F.~Igl\'oi, Z.~Szatm\'ari and Y.-C. Lin, Entanglement entropy dynamics of
  disordered quantum spin chains, {\em Phys. Rev. B\/} {\bf 85}, 094417-1--8
  (2012).

\bibitem[JK04]{MR2083138}
B.-Q. Jin and V.~E. Korepin, Quantum spin chain, {T}oeplitz determinants and
  the {F}isher--{H}artwig conjecture, {\em J. Stat. Phys.\/} {\bf 116}, 79--95
  (2004).

\bibitem[JSBS03]{JitomirskayaSchSt}
S.~Jitomirskaya, H.~Schulz-Baldes and G.~Stolz, {Delocalization in random
  polymer models}, {\em Commun. Math. Phys.\/} {\bf 233}, 27--48 (2003).

\bibitem[Kir08]{artRSO2008Kir}
W.~Kirsch, An invitation to random {S}chr\"odinger operators, in {\em Random
  {S}chr\"odinger operators\/}, Panor. Synth\`eses, vol.~25, Soc. Math. France,
  Paris, 2008, pp. 1--119, with an appendix by Fr{\'e}d{\'e}ric Klopp.

\bibitem[KLP03]{MR2025824}
W.~Kirsch, O.~Lenoble and L.~Pastur, On the {M}ott formula for the ac
  conductivity and binary correlators in the strong localization regime of
  disordered systems, {\em J. Phys. A\/} {\bf 36}, 12157--12180 (2003).

\bibitem[Kli06]{Klich06}
I.~Klich, Lower entropy bounds and particle number fluctuations in a {F}ermi
  sea, {\em J. Phys. A\/} {\bf 39}, L85--L91 (2006).

\bibitem[Laf16]{Laflorencie:2016kg}
N.~Laflorencie, Quantum entanglement in condensed matter systems, {\em Phys.
  Rep.\/} {\bf 646}, 1--59 (2016).

\bibitem[LSS14]{LeschkeSobolevSpitzer14}
H.~Leschke, A.~V. Sobolev and W.~Spitzer, Scaling of {R}\'enyi entanglement
  entropies of the free {F}ermi-gas ground state: A rigorous proof, {\em Phys.
  Rev. Lett.\/} {\bf 112}, 160403-1--5 (2014).

\bibitem[LSS17]{LeschkeSobolevSpitzer17}
H.~Leschke, A.~V. Sobolev and W.~Spitzer, Trace formulas for {W}iener-{H}opf
  operators with applications to entropies of free fermionic equilibrium
  states, {\em J. Funct. Anal.\/} {\bf 273}, 1049--1094 (2017).

\bibitem[LGP88]{MR1042095}
I.~M. Lifshits, S.~A. Gredeskul and L.~A. Pastur, {\em Introduction to the
  theory of disordered systems\/}, Wiley, New York, 1988.

\bibitem[Mot68]{Mott68PhilMag}
N.~F. Mott, Conduction in non-crystalline systems. i. {L}ocalized electronic
  states in disordered systems, {\em Phil. Mag.\/} {\bf 17}, 1259--1268 (1968).

\bibitem[Mot70]{Mott70PhilMag}
N.~F. Mott, Conduction in non-crystalline systems. iv. {A}nderson localization
  in a disordered lattice, {\em Phil. Mag.\/} {\bf 22}, 7--29 (1970).

\bibitem[MS16]{MR3582443}
R.~Movassagh and P.~W. Shor, Supercritical entanglement in local systems:
  counterexample to the area law for quantum matter, {\em Proc. Natl. Acad.
  Sci. USA\/} {\bf 113}, 13278--13282 (2016).

\bibitem[NSS13]{MR3088230}
B.~Nachtergaele, R.~Sims and G.~Stolz, An area law for the bipartite
  entanglement of disordered oscillator systems, {\em J. Math. Phys.\/} {\bf
  54}, 042110--1--24 (2013).

\bibitem[PF92]{pastfig1992random}
L.~Pastur and A.~Figotin, {\em Spectra of random and almost-periodic
  operators\/}, Grundlehren der Mathematischen Wissenschaften, vol. 297,
  Springer, Berlin, 1992.

\bibitem[PS14]{PhysRevLett.113.150404}
L.~Pastur and V.~Slavin, Area law scaling for the entropy of disordered
  quasifree fermions, {\em Phys. Rev. Lett.\/} {\bf 113}, 150404-1--5 (2014).

\bibitem[PS18a]{MR3744386}
L.~Pastur and V.~Slavin, The absence of the selfaveraging property of the
  entanglement entropy of disordered free fermions in one dimension, {\em J.
  Stat. Phys.\/} {\bf 170}, 207--220 (2018).

\bibitem[Pes03]{Peschel:2003gz}
I.~Peschel, Calculation of reduced density matrices from correlation functions,
  {\em J. Phys. A\/} {\bf 36}, L205--L208 (2003).

\bibitem[PS18b]{PfirschSobolev18}
B.~Pfirsch and A.~V. Sobolev, Formulas of {Szeg\H{o}} type for the periodic
  {S}chr\"{o}dinger operator, {\em Commun. Math. Phys.\/} {\bf 358}, 675--704
  (2018).

\bibitem[PY14]{PhysRevB.89.115104}
M.~Pouranvari and K.~Yang, Maximally entangled mode, metal-insulator
  transition, and violation of entanglement area law in noninteracting fermion
  ground states, {\em Phys. Rev. B\/} {\bf 89}, 115104-1--5 (2014).

\bibitem[RM09]{MR2566337}
G.~Refael and J.~E. Moore, Criticality and entanglement in random quantum
  systems, {\em J. Phys. A\/} {\bf 42}, 504010-1--31 (2009).

\bibitem[Sch12]{schmuedgen2012unbounded}
K.~Schm{\"u}dgen, {\em Unbounded self-adjoint operators on {H}ilbert space\/},
  Graduate Texts in Mathematics, vol. 265, Springer, Dordrecht, 2012.

\bibitem[Sim96]{Simon1996}
B.~Simon, Bounded eigenfunctions and absolutely continuous spectra for
  one-dimensional {S}chr\"{o}dinger operators, {\em Proc. Amer. Math. Soc.\/}
  {\bf 124}, 3361--3369 (1996).

\bibitem[Sob13]{MR3076415}
A.~V. Sobolev, Pseudo-differential operators with discontinuous symbols:
  {W}idom's conjecture, {\em Mem. Amer. Math. Soc.\/} {\bf 222}, no. 1043
  (2013).

\bibitem[Sob15]{Sobolev:2014ew}
A.~V. Sobolev, {W}iener{\textendash}{H}opf operators in higher dimensions: the
  {W}idom conjecture for piece-wise smooth domains, {\em Integr. Equ. Oper.
  Theory\/} {\bf 81}, 435--449 (2015).

\bibitem[Sto16]{Stolz:2019uc}
G.~Stolz, Aspects of the mathematical theory of disordered quantum spin chains,
  \eprint{arXiv:1810.05047}.

\bibitem[Wol06]{Wolf:2006ek}
M.~M. Wolf, Violation of the entropic area law for fermions, {\em Phys. Rev.
  Lett.\/} {\bf 96}, 010404-1--4 (2006).

\end{thebibliography}
\end{document}